\documentclass[11pt,a4]{article}

\usepackage{amsmath,amsthm,amsfonts,amssymb}
\usepackage[toc,page]{appendix}
\usepackage{fullpage}
\usepackage[showonlyrefs]{mathtools}
\usepackage{graphicx,tikz}
\usetikzlibrary{cd}
\usepackage{ifthen}
\usepackage{authblk}
\usepackage{makecell}
\usepackage{tcolorbox}
\usepackage{lipsum}
\usepackage{scalerel}
\tcbuselibrary{skins,breakable}
\usepackage{bm}
\usepackage{color}
\usepackage{authblk}
\usepackage{ytableau}

\usepackage[colorinlistoftodos,textsize=tiny]{todonotes}

\usetikzlibrary{positioning,arrows.meta}

\newboolean{ElectronicVersion}
\setboolean{ElectronicVersion}{true}

\ifthenelse{\boolean{ElectronicVersion}}{
	\usepackage[a4paper=true,pdftex,bookmarks,pagebackref,
	plainpages=false, pdfpagelabels=true
	]{hyperref}}{}

\makeatletter
\DeclareRobustCommand{\Udots}{%
	\vcenter{\offinterlineskip
		\halign{%
			\hbox to .8em{##}\cr
			\hfil.\cr\noalign{\kern.2ex}
			\hfil.\hfil\cr\noalign{\kern.2ex}
			.\hfil\cr}%
	}%
}
\makeatother

\usepackage{hyperref}
\hypersetup{
	bookmarksnumbered=true, 
	unicode=false, 
	pdfstartview={FitH}, 
	pdftitle={Tight Time-Space Lower Bounds for Collision Finding and Element Distinctness under Label Symmetry}, 
	pdfauthor={Fr\'ed\'eric Magniez and Sebastian Zur}, 
	pdfsubject={}, 
	pdfcreator={}, 
	pdfproducer={}, 
	pdfkeywords={}, 
	pdfnewwindow=true, 
	colorlinks=true, 
	linkcolor=blue, 
	citecolor=blue, 
	filecolor=blue, 
	urlcolor=blue 
}

\newcommand{\eq}[1]{\hyperref[eq:#1]{(\ref*{eq:#1})}}
\renewcommand{\sec}[1]{\hyperref[sec:#1]{Section~\ref*{sec:#1}}}
\newcommand{\thm}[1]{\hyperref[thm:#1]{Theorem~\ref*{thm:#1}}}
\newcommand{\lem}[1]{\hyperref[lem:#1]{Lemma~\ref*{lem:#1}}}
\newcommand{\cor}[1]{\hyperref[cor:#1]{Corollary~\ref*{cor:#1}}}
\newcommand{\app}[1]{\hyperref[app:#1]{Appendix~\ref*{app:#1}}}
\newcommand{\tabl}[1]{\hyperref[tab:#1]{Table~\ref*{tab:#1}}}
\newcommand{\defin}[1]{\hyperref[def:#1]{Definition~\ref*{def:#1}}}
\newcommand{\fig}[1]{\hyperref[fig:#1]{Figure~\ref*{fig:#1}}}
\newcommand{\clm}[1]{\hyperref[clm:#1]{Claim~\ref*{clm:#1}}}
\newcommand{\conj}[1]{\hyperref[conj:#1]{Conjecture~\ref*{conj:#1}}}
\newcommand{\probl}[1]{\hyperref[probl:#1]{Problem~\ref*{probl:#1}}}
\newcommand{\rem}[1]{\hyperref[rem:#1]{Remark~\ref*{rem:#1}}}
\newcommand{\para}[1]{\hyperref[para:#1]{Paragraph~\ref*{para:#1}}}
\newcommand{\exmp}[1]{\hyperref[exmp:#1]{Example~\ref*{exmp:#1}}}
\newcommand{\appx}[1]{\hyperref[appx:#1]{Appendix~\ref*{appx:#1}}}
\newcommand{\fct}[1]{\hyperref[fct:#1]{Fact~\ref*{fct:#1}}}

\newcommand{\Tr}{\operatorname{Tr}}
\newcommand{\id}{\mathrm{Id}}

\newcommand{\thmthm}[2]{\hyperref[thm:#1]{Theorem~\ref*{thm:#1}} and~\hyperref[thm:#2]{\ref*{thm:#2}}}
\newcommand{\lemlem}[2]{\hyperref[lem:#1]{Lemma~\ref*{lem:#1}} and~\hyperref[lem:#2]{\ref*{lem:#2}}}

{\endtcolorbox}

\newcommand{\nocontentsline}[3]{}
\newcommand{\tocless}[2]{\begingroup\let\addcontentsline\nocontentsline#1{#2}\endgroup}

\newtheorem{result}{Result}

\newtheorem{theorem}{Theorem}[section]
\newtheorem{lemma}[theorem]{Lemma}

\newtheorem{corollary}[theorem]{Corollary}
\newtheorem{fact}[theorem]{Fact}
\newtheorem{claim}[theorem]{Claim}

\newtheorem{remark}[theorem]{Remark}
\newtheorem{definition}[theorem]{Definition}

\usepackage{thm-restate}

\usepackage{color}
\definecolor{darkgreen}{rgb}{0,.5,0}
\definecolor{darkred}{rgb}{.7,.3,.3}
\definecolor{deepblue}{rgb}{0,.1,.7}

\newif\iflongversion
\newif\ifshortversion
\newif\ifediting
\editingtrue 

\def\ket#1{{\lvert}#1\rangle}
\def\bra#1{{\langle}#1\rvert}

\def\braket#1#2{{{\langle}#1\vert}#2\rangle}
\def\abs#1{\lvert #1 \rvert}

\newcommand{\norm}[1]{\lVert #1 \rVert}

\def\Tr{\mbox{Tr}}

\def\proj#1{\ket{#1}\bra{#1}}
\def\ketbra#1#2{\ket{#1}\bra{#2}}

\tikzset{
	mybox/.style = {
		rectangle,
		rounded corners=3mm,
		draw=blue!70!black,
		thick,
		fill=blue!5,
		minimum width=3cm,
		minimum height=1cm,
		align=center
	},
	arr/.style = {
		-{Straight Barb[scale=1]},
		line width=1pt,
		draw=blue!70!black
	}
}

\title{
Tight Time-Space Lower Bounds for Collision Finding \\
and Element Distinctness under Label Symmetry
}
\author{Fr\'ed\'eric Magniez\thanks{Email:\texttt{frederic.magniez@irif.fr}} } \author{Sebastian Zur\thanks{Email:\texttt{zursebastian@gmail.com}}}
\affil{CNRS, Université Paris Cité, IRIF, France}
\begin{document}
    \date{\today\\[1ex]\large}
    \pagenumbering{roman}
	\maketitle

\thispagestyle{empty}
\begin{abstract}
    How much memory is needed to retain the quantum speedup for collision finding? For a uniformly random function $f:[N]\to [N]$, the BHT algorithm~\cite{brassard1997collision} finds a collision using $O(N^{1/3})$ queries and a quantumly accessible classical table containing $O(N^{1/3})$ input-output pairs, whereas a logarithmic-space Grover search uses $O(\sqrt N)$ queries. Determining the optimal query-space tradeoff between these extremes remains a major open problem.

    We resolve this equation within the class of \emph{label-symmetric} algorithms, which treat the function $f$'s output labels as interchangeable. We prove that such algorithm that makes $T$ queries, uses $S$ qubits, and finds a collision in a uniformly random function $f:[M]\to [N]$ with constant probability satisfies
    \[
    T=\Omega(N^{1/3}) \qquad\text{and}\qquad T^2S=\Omega(N\log N).
    \]
    For the setting where $M=N$, these bounds are matched by a space-efficient implementation of the BHT algorithm. As a consequence of our tradeoff, any label-symmetric algorithm for the search version of Element Distinctness on $f: [n] \to [n^2]$ must satisfy
    \[
    T=\Omega(n^{2/3}) \qquad\text{and}\qquad T^2S=\Omega(n^2\log n),
    \]
    matching Ambainis's quantum walk~\cite{ambainis2004QWalkForElementDist}. Thus, both tradeoffs are optimal within the class of label-symmetric algorithms.

    To prove these results, we develop a space-sensitive version of the compressed oracle technique. The compressed oracle records the information learned by the algorithm in an evolving superposition of databases. Using label symmetry and representation theory, we show that an algorithm using $S$ qubits can effectively retain information about only $O(S/\log N)$ collision-free database entries. Substituting this estimate into the compressed oracle technique yields the stated tradeoffs. The key step identifies the relevant subspace of collision-free compressed databases with the least eigenspace of an arrangement graph, whose vertices are injective $s$-tuples and whose edges join tuples that differ in precisely a single coordinate. Of separate technical interest, we sharpen the previous analyses of the bottom of the spectrum of these arrangement graphs, identifying the exact spectral gap above the least eigenvalue.

\end{abstract}

\thispagestyle{empty}
\clearpage

\tableofcontents
\thispagestyle{empty}
\clearpage
\pagenumbering{arabic}

\section{Introduction}

\subsection{Context and motivation}

\paragraph{Post-quantum cryptography.}
One of the most dramatic potential applications of quantum computing is its ability to compromise widely used cryptographic tools. Shor's algorithm~\cite{shor1994Factoring}, for example, factors integers and computes discrete logarithms in quantum polynomial time. These problems form the main pillars of many of today's public-key cryptosystems, while the best-known classical algorithms for them require sub-exponential or exponential time. This threat has prompted an international effort in post-quantum cryptography to identify new fundamental primitives that would remain secure against quantum computers~\cite{PQCrypto2006,bernstein2009a}.

Note that an alternative approach to strengthen our current \emph{classical} cryptographic primitives is to shift to quantum technology and use quantum primitives such as the BB84 quantum key distribution~\cite{bb84}. Even though this technology is much more mature and currently being tested via satellites and telecom fibres, it still relies on classical cryptosystems in order, for instance, to authenticate classical channels of communication.

A less spectacular, but still serious threat comes from the BHT algorithm~\cite{brassard1997collision} for finding collisions in hash functions. Hash functions are also a pillar of modern symmetric-key cryptography. A basic security requirement for any hash-based cryptosystem is resistance to finding distinct inputs on which the hash function takes the same value, since many attacks begin by finding such a collision pair.
Even an ideal hash function is vulnerable to the generic randomised attack based on the birthday paradox, which gives a quadratic improvement over deterministic exhaustive search. Security parameters are therefore chosen with this attack in mind.
The BHT algorithm gives a further quantum improvement, finding a collision in a random function with a cubic-root quantum query complexity. Consequently, resisting quantum attacks requires larger security parameters than resisting classical attacks.

The quantum advantage for collision finding, however, strongly depends on the memory available to the algorithm. In particular, the BHT algorithm relies on a large memory with quantum access. This makes the role of memory in quantum collision-finding algorithms especially important.

\paragraph{Bottleneck of quantum memory.}
Quantum memory will probably be a major bottleneck for future quantum computers, especially when it must be accessed in quantum superposition in a single unit of time, as in a quantum analogue of random-access memory (RAM). Such memory has been formalised and used in several quantum algorithms, including the BHT algorithm and Ambainis's quantum walk algorithm for Element Distinctness~\cite{ambainis2004QWalkForElementDist}. In the former, the memory itself remains classical but allows quantum access and is therefore called QRACM; in the latter, the memory itself is quantum and is called QRAM (or a random-access gate). An early proposal and physical-architecture study of QRAM appears in~\cite{giovannetti2007QuantumRAM}.

In the Turing model, access to memory is sequential and thus takes linear time to address one cell of the memory. In the Random Access Memory (RAM) model, access is direct and can be done in basically one unit of time. The RAM model is closer to actual computers, but it has limitations when processing, for instance, massive data.
Note that in the circuit model, a multiplexer can implement the analogue of a RAM operation, but it costs linear circuit size, since it uses a linear number of Boolean gates with constant fan-in.

Thus, (quantum) polynomial-time algorithms usually do not explicitly require (Q)RAM, since they can simulate it with a polynomial overhead. Nonetheless, this remains critical for other algorithms, in particular those processing massive data. So algorithms such as Shor's algorithm do not explicitly require QRAM, but for those based on Grover Search~\cite{grover1996QSearch}, which provides a quadratic speed-up that is often significant for massive data, the quest for QRAM is fundamental.

Based on those facts, some criticisms were made against the BHT and Grover algorithms~\cite{bernstein2009b}. A 2D-grid architecture of $N$ classical processors with $\log N$ memory has been compared to a quantum computer with $\log N$ memory but with quantum access to a large classical memory (QRACM) of size $N$. In particular, it has been shown that the 2D-grid architecture proceeds as well as the quantum computer, and even better in some cases~\cite{bernstein2009b,jeffery11}.

To answer this criticism, another quantum algorithm was designed in~\cite{ChaillouxNS17} which uses only logarithmic-size quantum memory, but polynomial classical memory, while still providing a significant but smaller quantum advantage in its time complexity.

This illustrates the importance of understanding the role of memory in quantum algorithms. Since it is still unclear which memory model (QRAM, QRACM, or more restricted quantum memory) will be realistic in future quantum computers, it is natural to study how the time complexity depends on the available space, independently of the particular memory model. This leads to the study of time-space tradeoffs.

\paragraph{Time-space tradeoffs.}
Time-space tradeoffs quantify how the running time of an algorithm depends on the available memory. Their study has two complementary parts: algorithms provide upper bounds, while limitations on all algorithms provide lower bounds. Because unconditional lower bounds on computational time are generally out of reach, one needs further assumptions such as the exponential time hypothesis (ETH)~\cite{ImpagliazzoP99}, or, as in this work, one can use the query-complexity model, where access to the input is restricted to a specified oracle.
In that model, one input query is supposed to take one unit of time. Then, given some space-bound restriction, we prove a lower bound on the query complexity instead of the time complexity.

For instance, given comparison access to an array of size $N$, sorting the array requires
at least $\Omega(N \log N)$ (classical or quantum) comparisons~\cite{hoyerNS01}. If instead space is limited to at most $S$ bits, then the number of comparisons $T$ satisfies $TS=\Theta(N^2)$ (when $S=\Omega(\log N)$ and $S=O(N\log N)$)~\cite{BorodinFKLT81,PagterR98} for randomised algorithms.
In the case of quantum algorithms, one can show that $T^2S=\tilde\Theta(N^3)$~\cite{klauck2007quantum}.

Other such randomised and quantum tradeoffs have been established for a range of linear algebra problems including matrix-vector product, matrix inversion, matrix multiplication and powering~\cite{klauck2007quantum,BeameKW26}.

However, in the context of collision finding, the situation is much more complex.
First, in the classical evaluation model, when one has access to some random hash function $f:[N]\to[N]$ by querying $f(x)$ for any $x$, there exist methods, such as the parallel collision search algorithm~\cite{OorschotW99}, that succeed in finding a collision within $\tilde{O}(\sqrt{N})$ queries and only polylogarithmic memory, ruining the hope of any time-space tradeoff. Together with the $\Omega(\sqrt N)$ query lower bound for collision finding in a random function, even without a space restriction, this shows that additional memory cannot yield a substantial query-complexity improvement in this setting.

When the function is no longer random, the situation is quite different. This problem is known as Element Distinctness and consists of deciding whether the function $f$ is injective. The best-known algorithm satisfies the tradeoff $T^2 S=\tilde O(N^3)$~\cite{BeameCM13,ChenJWW22,LyuZ2023}.
However, the best lower bound is barely superlinear~\cite{ajtai2005EDlowerbound,BeameSSV03}. One can interpret this as a consequence of the short output size.
Still, in the comparison model, the situation is a bit easier, and by an amazing tour de force, a line of work has established the almost tight classical tradeoff of $TS=\Omega(N^{2-o(1)})$ (the upper bound $TS=\tilde O(N^2)$ comes from sorting algorithms).

Due to the inherent difficulty in establishing time-space tradeoff limitations for short-output problems, such as collision finding, the quest for such a tradeoff for finding a collision in a random function took some inspiring detours. One of them consists of finding $K$ collision pairs in a random hash function $f:[N]\to [N]$.
Classically, Dinur~\cite{Dinur20} proved the tight lower bound $T^2 S=\tilde{\Omega}(K^2 N)$, matching the upper bound of~\cite{OorschotW99}. The problem was subsequently considered in the quantum setting by Hamoudi and Magniez~\cite{hamoudi2020quantum}. Their work popularised the compressed oracle technique, also known as the recording queries technique, of Zhandry~\cite{zhandry2019record}, as a method for proving strong quantum query lower bounds in the low-success-probability regime, which in turn yield strong direct product bounds and time-space lower bounds. In particular, they proved $T^{3}S=\Omega(K^{3}N)$ and, by adapting the BHT algorithm, obtained $T^{2}S=\widetilde{O}(K^{2}N)$ throughout the range $\widetilde{\Omega}(\log N)\leq S\leq \widetilde{O}(K^{2/3}N^{1/3})$.

However, determining the optimal quantum time-space tradeoff for finding a \emph{single} collision remains today a major open problem, highlighted by Aaronson in~\cite{aaronson2021open} as quoting below: 
\begin{quote}
    \textbf{Open Problem} \cite[Problem 3]{aaronson2021open} {\it What are the optimal tradeoffs between the number of queries used by a quantum algorithm to solve the collision or the element distinctness problems, and the number of qubits or classical bits of memory?}
\end{quote}

\subsection{Contributions}

In this work, we solve the above open problem for a restricted class of quantum algorithms. We establish the first quantum time-space lower bound for finding a \emph{single} collision in a \emph{random} function $f:[M]\to[N]$, for the class of \emph{label-symmetric} algorithms, and show that the resulting tradeoff is tight in the standard case $M=N$.

Informally, an algorithm is \emph{label-symmetric} if its strategy does not depend on the labels assigned to the range elements: the strategy is similar for $f$ and $\sigma\circ f$, for any permutation $\sigma$ of $[N]$. The formal definition is stated in \defin{label-symmetric}.

\begin{result}[Informal version of {\thm{collision-time-space}}]
    Any label-symmetric algorithm that finds a collision in a uniformly random function $f:[M]\to[N]$ with constant probability and uses $T$ queries and $S$ qubits of space must satisfy
    \[
    T=\Omega(N^{1/3}) \qquad\text{and}\qquad T^2S=\Omega(N\log N).
    \]
\end{result}
It is natural to assume $M=\Omega(\sqrt N)$, in which case a uniformly random-function contains a collision with probability bounded away from zero. If $M=\omega(\sqrt N)$, this probability tends to one. Neither assumption is needed for the correctness of the lower bound.

Taking the domain size to be $n$ and the range size to be $n^2$ gives a tight consequence for the search version of Element Distinctness, because a uniformly random function $f:[n]\to[n^2]$ contains a collision with constant probability.

\begin{result}[Informal version of \cor{search-element-distinctness}]
    Any label-symmetric quantum algorithm that finds a collision in every non-injective function $f:[n]\to[n^2]$ with bounded error satisfies
    \[
    T=\Omega(n^{2/3})\qquad\text{and}\qquad T^2S=\Omega(n^2\log n).
    \]
\end{result}

It is natural to ask ourselves whether our restriction is too strong. In fact, both the uniform input distribution and the condition that an output pair forms a collision are already invariant under permutations of the range labels. The labels therefore carry no intrinsic meaning, and there is no evident reason for an algorithm to benefit from treating particular labels differently. 

This motivates our focus on label-symmetric algorithms. Indeed, both the BHT algorithm and Ambainis's quantum walk for Element Distinctness satisfy this condition (\lem{bht-label-symmetric} and \lem{bht-label-symmetric}). More generally, any algorithm whose only input-dependent operations are pairwise equality queries asking whether $f(x)=f(x')$ is automatically label-symmetric (\lem{equality-query-label-symmetric}).
\begin{result}[Informal versions of {\lem{bht-label-symmetric} and \lem{equality-query-label-symmetric}}]
    The following algorithms are label-symmetric:
    \begin{itemize}
        \item BHT algorithm for Collision Finding;
        \item Ambainis's quantum walk for Element Distinctness;
        \item Any algorithm whose only input-dependent operations are pairwise equality queries.
    \end{itemize}
\end{result}

For random collision finding with $M=N$, BHT with table size $r$ uses
\[
T=O\bigg(r+\sqrt{\frac Nr}\bigg)\qquad\text{and}\qquad S=O(r\log N).
\]
For $1\leq r\leq N^{1/3}$ this gives $T^2S=O(N\log N)$, while $r=N^{1/3}$ gives $T=O(N^{1/3})$. For Element Distinctness, Ambainis's quantum walk with $1\leq r\leq n^{2/3}$ gives $T^2S=O(n^2\log n)$. Thus, our tradeoffs for random collision finding and the search version of Element-Distinctness bounds are tight.

Our bounds also extend to algorithms with intermediate measurements, under a conditional version of label symmetry that we develop in \sec{intermediate-measurements}. This allows, for example, measuring a label and then using Grover search to find another occurrence of that label.

Of separate technical interest, we obtain a sharp analysis of the bottom of the spectrum of arrangement graphs. An arrangement graph may be viewed as an ordered analogue of a Johnson graph: its vertices are injective $s$-tuples $(y_1,\dots,y_s)\in[N]^s$, and two vertices are adjacent if and only if they differ in exactly one coordinate. We sharpen previous spectral results for these graphs~\cite{chen2013cyclic,araujo2017spectra} throughout the full regime $N\geq2s$, determining both the least eigenspace and the exact spectral gap above it. 
\begin{result}[Arrangement Spectral Gap, informal version of {\lem{spectrum}}]
    Assume $N\geq 2s$. The smallest eigenvalue of $A_{N,s}$ is $-s$, and the gap above the least eigenvalue is $N-2s+2$. 
\end{result}
This spectral estimate is a key ingredient in our space lower bound: it allows us to control the collision-free part of a compressed database that can be retained by an $S$-qubit label-symmetric algorithm. We give the precise spectral statement and explain its connection to the collision problem in the technical overview below.

\subsection{Related works}

The quest for quantum time-space tradeoff limitations for finding a collision for random hash functions was initiated for the lower bound part in~\cite{hamoudi2020quantum}, but for the case of finding multiple collisions only. The compressed oracle technique was introduced by Zhandry in~\cite{zhandry2019record}, mostly for studying the security of classical cryptosystems. Since then, its applications were developed for various settings including query complexity.

The quantum query complexity of symmetric algorithms has been studied before, for example for insertion into an ordered list~\cite{farhi1999} and for ordered search~\cite{carolan2025}, but these works use a different notion of symmetry (translation invariance). In our work, we consider the symmetry for which the algorithm’s strategy is unchanged when the range labels are permuted. This symmetry is more relevant for the collision finding problem.

\subsection{Open problems}

The main open problem is to remove the label-symmetry assumption. Because both the uniform input distribution and the collision success condition are invariant under range relabelling, it is difficult to imagine how treating particular labels asymmetrically could help. Nevertheless, the usual method of symmetrising an arbitrary algorithm stores a random permutation and can require $\Theta(N\log N)$ additional space, so it does not preserve the parameter that our lower bound tracks. A space-preserving symmetrisation argument, or a proof that avoids symmetry altogether, would extend the tradeoff to all algorithms.

\subsection{Technical overview}

The main difficulty in proving time-space tradeoffs for quantum collision finding is that the usual analysis only measures how much information the algorithm has learned, but does not directly capture how much of this information can be stored in a limited amount of quantum memory. Our proof overcomes this difficulty by combining three ingredients: the compressed oracle technique, the label symmetry of algorithms that treat output labels as interchangeable, and a representation-theoretic analysis of the linear spaces in the aforementioned compressed oracle technique. Together, these tools show that an algorithm using only $S$ qubits cannot maintain a large collision-free database, which limits how quickly it can create a collision.

\paragraph{Compressed-oracle progress.} Our proof is based on the compressed oracle technique~\cite{zhandry2019record,hamoudi2020quantum}, which we introduce in \sec{compressed-oracle}. In this technique, we represent what the algorithm knows about a uniformly random function $f:[M]\to[N]$, by a superposition of \textit{databases}. These databases are initially filled with ''empty`` cells $(x,\bot)$, representing that the algorithm has no information about the value $f(x)$, but after $t$ queries, each branch of the superposition contains a database with at most $t$ cells of the form $(x,f(x))$. The unitary evolution of these database branches through queries is quite subtle, but it was formally shown in the seminal work by Zhandry~\cite{zhandry2019record} that we can think of a query (to $x$) adding a random fresh cell $(x,f(x))$ to the database.

The progress of the algorithm can now be tracked by defining projections onto these databases, such as $\Pi_{\geq 1}$, which projects onto all branches such that there exist entries $(x_1,y),(x_2,y)$, i.e. a collision, or $\Pi_{=0}$ which projects onto all databases that contain no collision. More formally, if $\ket{\widetilde\psi_t}$ is the superposition of our databases after $t$ queries,
\[
\Delta_t:=\norm{\Pi_{\geq1}\ket{\widetilde\psi_t}}
\]
denotes the amplitude on such databases that contain at least one collision. A new query can create a collision only by adding new cell whose label $f(x)$ is already present at another occupied position. This intuition allows one to bound the increase of $\Delta_t$, formalised in \lem{one-query}, in terms of the number of occupied entries in the collision-free part of the database. Write $\Lambda_s$ for the projection onto databases with exactly $s$ occupied positions, then
\[
\Delta_{t+1}\leq\Delta_t+\frac4{\sqrt N}\bigg(\sum_{s=0}^t s\norm{\Pi_{=0}\Lambda_s\ket{\widetilde\psi_t}}^2\bigg)^{1/2}.
\]

Since $s$ is always bounded by $t$ and that $\{\Pi_{=0}\Lambda_s\}_s$ form a set of mutually orthogonal projections, using Cauchy-Schwarz we obtain the standard estimate $\Delta_T=O(T^{3/2}/\sqrt N)$. Since the probability of success of our algorithm is approximately equal to $\Delta_T$, we recover the usual $T=\Omega(N^{1/3})$ quantum query lower bound, if we require that $\Delta_T = \Omega(1)$. Our main task is to replace this crude bound by one that also depends on the available space.

\paragraph{The space restriction from label symmetry.} In \sec{symmetry-restrictions}, we introduce label symmetry at the level of the reduced state of the input register $\rho_I$, obtained by tracing out the algorithm's register in the joint state on the algorithm and the input.  Label symmetry implies that the support of $\rho_I$ is invariant under permutations of the $N$ range labels. However, any quantum algorithm that uses at most $S$ qubits, the Schmidt rank across the joint state is at most $2^S$, and hence the reduced state $\rho_I$ has a rank of at most $2^S$. 
Those two observations imply that the orbit (under the symmetric group $\mathfrak S_N$) of every vector $\ket{\beta}$ in the support of $\rho_I$ spans a space of dimension at most $2^S$, leading to the following \emph{orbit bound} that we state formally in \lem{orbit-span-symm}:
\[
\dim\operatorname{span}\{\sigma\ket{\beta}:\sigma\in\mathfrak S_N\}\leq2^S.
\]

In \sec{label-symmetric}, we verify that this label-symmetry assumption is satisfied by both the BHT collision-finding algorithm and Ambainis's quantum walk for Element Distinctness, as well as by the broader class of equality-query algorithms. This shows that label symmetry captures natural quantum algorithms and, together with the BHT upper bound, establishes the tightness of our tradeoff in the standard case $M=N$.

\paragraph{Collision-free compressed databases.} To exploit the orbit bound, in \sec{no-col} we introduce a representation-theoretic model for a database with $s$ occupied positions. After fixing these positions and omitting all $\bot$ entries, the recorded values may be identified with vectors in $\mathbb C[N]^{\otimes s}$. Two subspaces of this tensor space arise naturally.

The first is
\[
W:=\mathrm{span}\{\ket{\widehat{0}}\}^{\perp} \subseteq \mathbb C[N], \qquad \ket{\widehat{0}}:=\frac1{\sqrt N}\sum_{y\in[N]}\ket y.
\]
Each $\ket{\widehat{0}}$ intuitively represents `knowing nothing' about the function value, and is mapped to a $\bot$ cell in the compressed oracle technique. Thus, valid databases with $s$ non-$\bot$ cells lie precisely in $W^{\otimes s}$.

The second is the collision-free subspace
\[
C_s^\perp:=\operatorname{span}\{\ket{y_1,\ldots,y_s}:y_1,\ldots,y_s\text{ are pairwise distinct}\}.
\]
Consequently, the central space in our analysis is
\[
C_s^\perp\cap W^{\otimes s},
\]
which consists of states that are simultaneously collision-free and compatible with a database in the compressed-oracle representation.

\paragraph{Representation theory and arrangement graphs.} We next determine the $\mathfrak S_N$-representation carried by our space $C_s^\perp\cap W^{\otimes s}$. To achieve this, we define the deletion map $D:=\bigoplus_{r=1}^s d_r|_{C_s^\perp}$, obtained by, for each coordinate $r\in\{1,\ldots,s\}$, applying $d_r$ which deletes the $r$-th coordinate. On the one hand (\clm{kernel}), this map $D$ is related to $C_s^\perp\cap W^{\otimes s}$ through the identity
\[
\ker D=C_s^\perp\cap W^{\otimes s}.
\]

On the other hand, the operator $D^\dagger D$ has a natural graph-theoretic interpretation,
\[
D^\dagger D=A_{N,s}+s\cdot\id,
\]
where $A_{N,s}$ is the adjacency operator of the arrangement graph. This graph may be viewed as an ordered analogue of a Johnson graph: its vertices are the injective $s$-tuples in $[N]^s$, and two vertices are adjacent when they differ in exactly one coordinate. It follows that the space of collision-free compressed databases is precisely the $(-s)$-eigenspace of the arrangement graph:
\[
C_s^\perp\cap W^{\otimes s}=\ker(A_{N,s}+s\cdot\id),
\]

Our spectral analysis, stated in \lem{spectrum} and proved in \sec{proof-spectrum}, shows, for every $N\geq2s$, the space $\ker(A_{N,s}+s\cdot\id)$ decomposes into irreducible representations that are all of high dimension, namely at least
\[
\binom Ns-\binom N{s-1}.
\]
This exceeds $2^S$ whenever $s>2S/\log_2N$, meaning that a vector accessible to an $S$-qubit label-symmetric algorithm cannot have a nonzero component in $\ker(A_{N,s}+s\cdot\id) = C_s^\perp\cap W^{\otimes s}$ for such large values of $s$. This will be crucial for establishing our time-space tradeoff.

\paragraph{The noncommuting projections difficulty.} One difficulty is that the two natural conditions on a database, being collision-free and being a valid compressed database, are described by projections in the computational basis and the Fourier basis, respectively, and hence do not commute. The spectral gap of the arrangement graph controls precisely this discrepancy. In \lem{proj-diff}, we prove
\begin{equation}\label{eq:intro-noncom}
\norm{\Pi_{C_s^\perp}\Pi_{W^{\otimes s}}-\Pi_{C_s^\perp\cap W^{\otimes s}}}^2\leq\frac{2s-2}{N}.
\end{equation}
Consequently, once the component in $C_s^\perp\cap W^{\otimes s}$ has been ruled out, the squared norm of the projection onto the collision-free subspace $C_s^\perp\Pi_{W^{\otimes s}}$ is at most $(2s-2)/N$ times the squared norm of the original state.

\paragraph{From the database space bound to the tradeoff.} In \sec{database-size}, we combine all our aforementioned intermediary results to obtain a tighter bound, one that involves the space $S$, on the quantity
\[
\sum_{s=0}^t s\norm{\Pi_{=0}\Lambda_s\ket{\widetilde\psi_t}}^2.
\]

We bound the quantity $s\norm{\Pi_{=0}\Lambda_s\ket{\widetilde\psi_t}}^2$, depending on whether $s$ is larger than $2S/\log_2N$ or not. If $s$ is smaller, the bound is straightforward and we obtain
\[
s\norm{\Pi_{=0}\Lambda_s\ket{\widetilde\psi_t}}^2 \leq \min\left\{t,\frac{2S}{\log_2N}\right\}\norm{\Lambda_s\ket{\widetilde\psi_t}}^2.
\]

For $s>2S/\log_2N$, we know from our representation-theoretic analysis of the space $C_s^\perp\cap W^{\otimes s}$, that $\ket{\widetilde\psi_t}$ can not have any overlap with $C_s^\perp\cap W^{\otimes s}$. Note that \eq{intro-noncom} implies that, for any normalised vector $v$ orthogonal to $C_s^\perp\cap W^{\otimes s}$, we have
\[
\norm{\Pi_{C_s^\perp}\Pi_{W^{\otimes s}}v}^2\leq\frac{2s-2}{N}.
\]
Therefore, by formally identifying the projection $\Pi_{=0}$ with the projection $\Pi_{C_s^\perp}\Pi_{W^{\otimes s}}$ in \clm{intertwine}, we can bound
\[
s\norm{\Pi_{=0}\Lambda_s\ket{\widetilde\psi_t}}^2 \leq\frac{2s-2}{N}s\norm{\Lambda_s\ket{\widetilde\psi_t}}^2.
\]

Assuming that $t \leq \sqrt{N}$ and $S \geq \log_2N$, this obtain the space-sensitive bound
\[
\Delta_{t+1}\leq\Delta_t+\frac4{\sqrt N}\bigg(\sum_{s=0}^t s\norm{\Pi_{=0}\Lambda_s\ket{\widetilde\psi_t}}^2\bigg)^{1/2} \leq \Delta_t+\frac4{\sqrt N}\bigg(\min\left\{t,\frac{2S}{\log_2N}\right\}\bigg)^{1/2},
\]
which for $\Delta_T = \Omega(1)$ and  $T=\Omega(N^{1/3})$ requires $T^2S=\Omega(N\log N)$, proving \thm{collision-time-space}.

\section{Preliminaries}

\subsection{Linear algebra}

For a positive integer $n$, we write $[n]:=\{0,\dots,n-1\}$. We consider finite-dimensional complex inner product spaces ${\cal H} = \mathbb{C}^d$ for some dimension $d$. We use standard bra-ket notation for column and row vectors in $\mathbb{C}^d$.  We consider all bra-ket vectors to be normalised unless specified otherwise. For a finite set $S$, we let
$$\mathbb{C}[S] := \mathrm{span}\{\ket{s}:s\in S\},$$
and, for a positive integer $n$, we abbreviate $\mathbb C[[n]]$ as $\mathbb C[n]$.
For any two Hermitian operators $A, B$, we write $A\succeq B$ if their difference $A - B$ is positive semidefinite.
\begin{definition}[Spectral norm]
    Let $A \in \mathbb{C}^{d \times d}$ be a matrix. Then the \emph{spectral norm} (also known as the operator norm) of $A$ is 
    $$\norm{A} := \sup\limits_{\ket{v} \in \mathbb{C}^{d}} \norm{A\ket{v}},$$
    where $\norm{A\ket{v}}$ is the standard vector $\ell_2$-norm.
\end{definition}
Since we consider all bra-ket vectors to be normalised, the above supremum is implicitly over normalised vectors.

\begin{definition}[Fourier basis]\label{def:fourier}
    Let $\{\ket{y}\}_{y \in [N]}$ be the computational basis for $\mathbb{C}[N]$. Then $\{\ket{\widehat{y}}\}_{y \in [N]}$ is the \emph{Fourier basis} of $\mathbb{C}[N]$, where each $\ket{\widehat{y}}$ is defined as
    $$ \ket{\widehat{y}} := \frac{1}{\sqrt{N}} \sum_{z \in [N]} \omega_N^{-yz}\ket{z}.$$
    Here $\omega_N = e^{\frac{2\pi\iota}{N}}$, where $\iota$ denotes the imaginary unit to prevent ambiguity with the variable $i$.
\end{definition}

\subsection{Representation theory of the symmetric group}

We introduce the representation-theoretic preliminaries of the symmetric group that are used in the rest of this section. 

Let $\mathfrak S_n$ denote the symmetric group of degree $n$. Throughout this subsection, all representations are finite-dimensional complex representations, and group actions are on the left unless stated otherwise. A \textit{representation} $(V,\rho)$  of $\mathfrak S_n$ is a complex vector space $V$ equipped with a homomorphism 
\[ 
\rho:\mathfrak S_n\longrightarrow \mathrm{GL}(V). 
\]

When the representation homomorphism is clear from context, we suppress it from the notation and identify the representation $(V,\rho)$ with $V$. In this case, we write $\sigma v$ for $\rho(\sigma)v$.

A \textit{subrepresentation} $W$ of $V$ is given by a subspace $W\subseteq V$ such that 
\[
\sigma w\in W
\] 
for every $\sigma\in\mathfrak S_n$ and $w\in W$. A representation $V$ is called \textit{irreducible} if its only subrepresentations are $\{0\}$ and $V$ itself. 

Let $(V,\rho_V)$ and $(U,\rho_U)$ be representations of $\mathfrak S_n$. A linear map $B:V\longrightarrow U$ is $\mathfrak S_n$-\emph{equivariant} if 
\[
B(\rho_V(\sigma) v)=\rho_U(\sigma) B(v)
\] 
for every $\sigma\in\mathfrak S_n$ and $v\in V$.

The \textit{character} of a representation $(V,\rho)$ is the function
\[
\chi_{V}:\mathfrak S_n\to\mathbb C, \qquad \chi_{V}(\sigma):=\Tr(\rho(\sigma)).
\]
Thus,
\[
\chi_{V}(1)= \Tr(\rho(1)) = \Tr(\id_V) = \dim V,
\]
where $1\in\mathfrak S_n$ denotes the identity element.

The irreducible complex representations of $\mathfrak S_n$ are indexed by \textit{partitions} $\lambda$ of the integer $n$, denoted by $\lambda\vdash n$, that is, weakly decreasing sequences of positive integers
\[
\lambda=(\lambda_1,\dots,\lambda_r),
\qquad
\abs{\lambda}:=\lambda_1+\cdots+\lambda_r=n.
\]
We denote the corresponding irreducible representation by $S^\lambda$, and call it the \textit{Specht module} of shape $\lambda$. Its character is denoted by 
\[
\chi_\lambda:=\chi_{S^\lambda}.
\]
Thus, in particular,
\[
\chi_\lambda(1)=\dim S^\lambda.
\]

Each partition $\lambda$ has an associated \textit{Young diagram} $Y(\lambda)$, a collection of left-aligned rows of boxes with $\lambda_i$ boxes in row $i$. For example, the partition $(4,2,1)\vdash 7$ has Young diagram
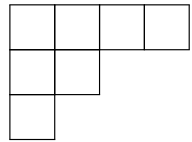
\begin{equation*}
    \begin{ytableau}
    \, & \, & \, & \, \\
    \, & \, \\
    \,
    \end{ytableau}
\end{equation*}
We also regard the empty partition $\varnothing$ as the unique partition of $0$, with $\abs{\varnothing}=0$ and $Y(\varnothing)=\varnothing$.

A \textit{standard Young tableau} of shape $\lambda$, denoted by the symbol $\mathfrak{t}$, is a filling of the boxes of $Y(\lambda)$ with the numbers $1,\dots,n$, each used exactly once, such that the entries increase along rows and down columns. For instance, a standard Young tableau of shape $(4,2,1)$ is
\begin{equation*}
    \begin{ytableau}
    1 & 2 & 4 & 7 \\
    3 & 5 \\
    6
    \end{ytableau}
\end{equation*}
since the entries increase from left to right in each row and from top to bottom in each column. We write $\mathrm{SYT}(\lambda)$ for the set of standard Young tableaux of shape $\lambda$.

\begin{theorem}[Theorem~$2.5.2$ in~\cite{sagan2001symmetric}]
\[
   \abs{\mathrm{SYT}(\lambda)} = \dim S^\lambda.
\]
\end{theorem}

\subsection{Quantum query complexity}\label{sec:query}

In this work, the input is a function $f:[M]\to[N]$. The memory of a quantum algorithm ${\cal A}$ is described, without loss of generality, by registers ${\cal W}$, ${\cal X}$, and ${\cal Y}$. The input oracle acts on ${\cal X}\otimes{\cal Y}$, while ${\cal W}$ is an additional workspace register. The algorithm accesses $f\in[N]^M$ through the following oracle.

\begin{definition}[Oracle]\label{def:query-og}
    An \emph{oracle} ${\cal O}_f$, encoding the input function $f \in [N]^M$, is a unitary transformation that acts on
    $$ \mathrm{span}\{\ket{x}_{\cal X}\ket{y}_{\cal Y} : x \in [M], y \in [N]\}, $$
    with its action on the basis state $\ket{x}_{\cal X}\ket{y}_{\cal Y}$ defined as
    $$ {\cal O}_f\ket{x}_{\cal X}\ket{y}_{\cal Y} = \ket{x}_{\cal X}\ket{(y + f(x))\bmod N}_{\cal Y}. $$
\end{definition}

The input $f$ is typically drawn from some (hard) input distribution $\delta$ on $[N]^M$, denoted by $f \sim \delta$. Consequently, ${\cal O}_f$ is a random variable. In adversary methods and the compressed-oracle technique, this randomness is purified by introducing an additional \emph{input} register $\mathcal{I}$ that stores a superposition of function tables. If $f \sim \delta$, the register $\mathcal{I}$ is initialised as
$$
\ket{\delta} := \sum_{f \in [N]^M} \sqrt{\delta(f)} \ket{f}_{\cal I}.
$$
Here, $\ket{\delta}$ represents the initial state of the input register. It is important to note that this should not be confused with the initial state of the algorithm, which is the all-zero state. This purification of the input leads to the following purified oracle:
\begin{definition}[Purified Oracle]\label{def:query}
    A \emph{purified oracle} ${\cal O}$ is a unitary transformation that acts on
    $$ \mathrm{span}\{\ket{x}_{\cal X}\ket{y}_{\cal Y}\ket{f}_{\cal I} : x \in [M], y \in [N], f \in [N]^M\}, $$
    with its action on the basis state $\ket{x}_{\cal X}\ket{y}_{\cal Y}\ket{f}_{\cal I}$ defined as
    $$ {\cal O}\ket{x}_{\cal X}\ket{y}_{\cal Y}\ket{f}_{\cal I} = \ket{x}_{\cal X}\ket{(y + f(x))\bmod N}_{\cal Y}\ket{f}_{\cal I}. $$
\end{definition}

From the perspective of the algorithm, it is indistinguishable whether it interacts with the random variable ${\cal O}_f$ or the purified oracle ${\cal O}$ with input register initialised to $\ket{\delta}$. The relationship between the two is captured by the following expression:
$$
{\cal O} = \sum_{f \in [N]^M} {\cal O}_f \otimes \proj{f}_{\cal I}.
$$

It is equivalent, and in this work more convenient, to encode the query into the phase by viewing the ${\cal Y}$ register in the Fourier basis $\{\ket{\widehat{y}}\}_{y \in [N]}$ instead of the computational basis $\{\ket{y}\}_{y \in [N]}$.
In this Fourier basis, the oracle from \defin{query} acts on any basis state $\ket{x}_{\cal X}\ket{\widehat{y}}_{\cal Y}\ket{f}_{\cal I}$ as
$$ {\cal O}\ket{x}_{\cal X}\ket{\widehat{y}}_{\cal Y}\ket{f}_{\cal I} = \omega_N^{yf(x)}\ket{x}_{\cal X}\ket{\widehat{y}}_{\cal Y}\ket{f}_{\cal I}. $$

\begin{definition}[$T$-Query Quantum Algorithm]\label{def:algorithm}
    A \emph{$T$-query quantum algorithm} ${\cal A}$ on $[N]^M$ is a sequence of unitaries $U_0,\dots,U_T$ acting on
    \[
        {\cal H}_{\cal WXY} = {\cal H}_{\cal W} \otimes \mathbb C[M] \otimes \mathbb C[N],
    \]
    where ${\cal H}_{\cal W}$ is an arbitrary finite-dimensional  workspace.

    For a fixed input $f:[M]\to[N]$ and $t\in[T+1]$, define
    \[
        \ket{\psi_t^f({\cal A})} :=  U_t{\cal O}_fU_{t-1}{\cal O}_f\cdots {\cal O}_fU_0\ket{0}_{\cal WXY}.
    \]
    Thus, for $t<T$, $\ket{\psi_t^f({\cal A})}$ is the state immediately before the $(t+1)$-st query, while    $\ket{\psi_T^f({\cal A})}$ is the final state of the algorithm.

    For an input distribution $\delta$ on $[N]^M$, define the corresponding purified joint state by
    \[
        \ket{\psi_t({\cal A},\delta)} := \sum_{f\in[N]^M}\sqrt{\delta(f)}\ket{\psi_t^f({\cal A})}_{\cal WXY}\ket f_{\cal I}.
    \]
    Equivalently,
    \[
        \ket{\psi_t({\cal A},\delta)} = U_t{\cal O}U_{t-1}{\cal O}\cdots{\cal O}U_0\ket{0}_{\cal WXY}\ket{\delta}_{\cal I}.
    \]

    Finally, we define the reduced state of the input register,
    \[
        \rho_{\cal I}^t({\cal A},\delta) := \Tr_{\cal WXY}\left[\ket{\psi_t({\cal A},\delta)}\bra{\psi_t({\cal A},\delta)}\right].
    \]
\end{definition}
In the definition of the joint state $\ket{\psi_t({\cal A},\delta)}$, the unitaries $U_0,\dots,U_t$ act on a larger Hilbert space than originally defined, but each operator is implicitly understood to be tensored with the identity operator on ${\cal I}$.

Observe from the definitions in \defin{algorithm} that the reduced state of the input register can alternatively be written as 
\begin{equation}\label{eq:input-gram}
    \rho_{\cal I}^t({\cal A},\delta)=\sum_{f,g \in [N]^M}\sqrt{\delta(f)\delta(g)}\braket{\psi_t^g({\cal A})}{\psi_t^f({\cal A})}\ketbra{f}{g}.
\end{equation}

In this work, we consider search problems for which an input may have multiple valid outputs, or possibly no valid output at all. We therefore use \textit{average-case} quantum query complexity, defined with respect to an input distribution $\delta$, rather than the \textit{worst-case} quantum query complexity.
\begin{definition}[$\epsilon$-error Average-Case Quantum Query Complexity]\label{def:complex-gen}
    Let $\Sigma$ be a finite set of possible outputs, and let ${\sf F}:[N]^M\rightarrow 2^\Sigma$ be a search problem, where ${\sf F}(f)\subseteq\Sigma$ denotes the set of valid outputs on input $f$.

    For an input distribution $\delta$ on $[N]^M$, the \emph{$\epsilon$-error average-case quantum query complexity} of ${\sf F}$ with respect to $\delta$, is the minimum number of queries needed by any quantum query algorithm ${\cal A}$ such that 
    \[
    \Pr_{f\sim\delta} \bigl[{\cal A}\text{ outputs some }z\in{\sf F}(f)\bigr] \geq 1-\epsilon.
    \]
    Here the probability is over both the choice of $f\sim\delta$ and the measurement outcomes of ${\cal A}$.
\end{definition}

This choice of metric is discussed in depth in~\cite{jeffery2025compressed}. For collision finding, without the promise that the input contains a collision, no algorithm can solve the corresponding search relation on every input. This is why we work with average-case complexity under the uniform input distribution.

For collision finding, we take 
\[
\Sigma_{\rm coll} := \{(x_1,x_2,y)\in[M]^2 \times [N]:x_1<x_2\}
\]
and define
\begin{equation}\label{eq:collision-relation}
    {\sf Coll}(f) := \{(x_1,x_2,y)\in\Sigma_{\rm coll}:f(x_1)=f(x_2)=y\}.
\end{equation}
Thus, ${\sf Coll}(f)$ is the set of valid collision certificates of $f$, and it may be empty if $f$ is injective.

This formulation is asymptotically equivalent, up to one additional query and $\lceil\log_2N\rceil$ additional qubits, to the usual formulation where the algorithm only has to output the collision pair $(x_1,x_2)$, as the algorithm may query $f(x_1)$ and output the resulting common value $y=f(x_1)$. Thus the two formulations have the same asymptotic query and space complexities. This alternative formulation is more convenient in this work  because it interfaces directly with the compressed oracle readout bound in \lem{collision-success}.

\subsection{Space complexity}
In this work, by quantum \emph{time-space lower bounds} we mean lower bounds that relate the \emph{query complexity} $T$ of an algorithm, as in \defin{algorithm}, to its \emph{space complexity} $S$. Such bounds show that an algorithm solving the problem cannot simultaneously use few queries and little space.

This reflects the standard convention in the query-complexity model that each query takes one unit of time. Consequently, a query lower bound also yields a time lower bound, relativised to oracle access to the input. 

The space complexity $S$ is the number of qubits in the algorithm's workspace registers (i.e.\ the qubits on which the circuit operates throughout the computation). Equivalently, since ${\cal H}_{\cal WXY}$ denotes the algorithm's Hilbert space, 
\[
\dim({\cal H}_{\cal WXY}) \leq 2^S.
\]

Given our space constraints, the principle of deferred measurements cannot be invoked, since the usual coherent simulation that retains measurement outcomes in additional registers may increase the workspace. The recent literature on this topic~\cite{GirishR22,Zhandry24} shows that it is possible to differ intermediate measurements at the cost of either a larger space complexity or time complexity, and do not justify assuming that intermediate measurements can be postponed while preserving both our query and space bounds.

For simplicity, we first assume in this paper that all transformation are unitaries and that the algorithms makes no intermediate measurements. In \sec{intermediate-measurements}, we extend our proofs to intermediate measurements, under a symmetry condition on the measurements themselves.

\section{The compressed oracle technique for collision finding}\label{sec:compressed-oracle}

\subsection{The compressed oracle technique}
In the compressed oracle technique~\cite{zhandry2019record}, also known as the recording query technique, the input distribution $\delta$ is initialised to the uniform distribution over all functions from $[M]$ to $[N]$, which we denote by ${\sf U}$:
\begin{equation}\label{eq:deltacomp}
    \ket{{\sf U}}_{\mathcal{I}} \coloneq \frac{1}{\sqrt{N^M}} \sum_{f \in [N]^M} \ket{f}_{\mathcal{I}} = \bigotimes_{x\in[M]}\bigg(\frac{1}{\sqrt N}\sum_{y\in[N]}\ket y_{{\cal I}_x}\bigg).
\end{equation}
This construction also extends to product distributions, although we do not need that generality here. Such an extension appears in~\cite{hamoudi2020quantum}. For background on the technique, see~\cite{zhandry2019record,hamoudi2020quantum,chung2021compressed}. 

The input register ${\cal I}$ holding a computational basis state $\ket{f}_{\cal I}$, where $f \in [N]^M$, is the tensor product of the function values of $f$ for the different values of $x \in [M]$:
\[\ket{f}_{\cal I} = \bigotimes_{x \in [M]}\ket{f(x)}_{{\cal I}_x}.\]

We enlarge the Hilbert space of every cell ${\cal I}_x$ to $\mathbb C[[N] \cup \{\bot\}]$. We continue to write $\ket f$ for a computational basis state in $\mathbb C[[N] \cup \{\bot\}]^{\otimes M}$, so that $f(x)$ may equal $\bot$. For every $x \in [M]$, define the isometry
\begin{align*}
    \mathsf{Comp}_x:\mathbb C[N]\longrightarrow\mathbb C[[N] \cup \{\bot\}],\qquad \mathsf{Comp}_x :=\ket{\bot}\bra{\widehat{0}} +\sum_{z \in [N] \setminus \{0\}}\proj{\widehat z}.
\end{align*}
Thus, the uniform state $\ket{\widehat0}$, which represents having no information about the value in cell ${\cal I}_x$, is mapped to $\ket\bot$, while every nonzero Fourier state remains unchanged. In particular,
\[
\mathsf{Comp}_x^\dagger\mathsf{Comp}_x=\id_{\mathbb C[N]}, \qquad \mathsf{Comp}_x\mathsf{Comp}_x^\dagger =\id_{\mathbb C[[N] \cup \{\bot\}]}-\proj{\widehat0}.
\]
Taking the tensor product over $x \in [M]$ and extending it by the identity on the algorithm registers gives the isometry
\begin{equation}\label{eq:comp}
    \mathsf{Comp}:= \id_{\cal WXY} \otimes \bigotimes_{x \in [M]} \mathsf{Comp}_x.
\end{equation}

Let $\widetilde{\cal O}$ be the unitary \textit{recording query operator} obtained by extending the restriction of $\mathsf{Comp}{\cal O}\mathsf{Comp}^{\dagger}$ to $\operatorname{im}(\mathsf{Comp})$ locally: it acts only on the queried database cell and acts as the identity when the ${\cal Y}$ register is $\ket{\widehat0}$. Its explicit basis action is given in \lem{oracle}. This extension preserves $\operatorname{im}(\mathsf{Comp})$ and satisfies
\begin{equation}\label{eq:recording-prop}
   \widetilde{\cal O}\Pi_{\sf Comp} = \Pi_{\sf Comp}\widetilde{\cal O} =\mathsf{Comp}{\cal O}\mathsf{Comp}^\dagger, \qquad \widetilde{\cal O}\mathsf{Comp}=\mathsf{Comp}{\cal O},
\end{equation}
where $\Pi_{\sf Comp}:=\mathsf{Comp}\mathsf{Comp}^\dagger$ is the orthogonal projector onto $\operatorname{im}(\mathsf{Comp})$; its explicit form appears in \eq{image-iso}.

\begin{remark}
    Our formulation differs slightly from the recording-query formalism of Hamoudi and Magniez~\cite{hamoudi2020quantum}. There, the corresponding change of representation is implemented by a unitary on the enlarged space $\mathbb C[[N]\cup\{\bot\}]$: locally, it exchanges $\ket{\bot}$ with $\ket{\widehat 0}$ and fixes every $\ket{\widehat z}$ with $z\neq0$. The map $\mathsf{Comp}_x$  above is precisely the restriction of this unitary to the original subspace $\mathbb C[N]$. Thus, the two formulations give equivalent descriptions of the recording query dynamics.

    We instead follow the isometric viewpoint of Jeffery and Zur~\cite{jeffery2025compressed}, in which $\mathsf{Comp}$ maps the standard input space into the enlarged database space and one explicitly keeps track of its image $\operatorname{im}(\mathsf{Comp})$. Although the two viewpoints are equivalent at the level of the recording query dynamics, making this
    image explicit is convenient for our space-sensitive analysis. In particular, all compressed states remain in $\operatorname{im}(\mathsf{Comp})$, and later, after restricting to the occupied database cells, this image condition becomes the Fourier-space constraint that enters our representation-theoretic analysis of collision-free databases.
\end{remark}

In the compressed oracle technique, we study a modified version of the joint state $\ket{\psi_t({\cal A},\delta)}$. Since the rest of this work only considers the case $\delta={\sf U}$, we omit the distribution parameter and simply write $\ket{\psi_t({\cal A})}$. In this modified state, each oracle call is replaced by $\widetilde{\cal O}$ and the input is initialised to $\ket{\bot^M}$ instead of $\ket{\sf U}$:
\[ 
\ket{\widetilde{\psi}_t({\cal A})} = U_t\widetilde{\cal O}U_{t-1}\widetilde{\cal O}\dots \widetilde{\cal O}U_0\ket{0}_{\cal WXY}\ket{\bot^M}_{\cal I}.
\] 
Note that since $\ket{0}_{\cal WXY}\ket{\bot^M}_{\cal I}$ lies in $\operatorname{im}(\mathsf{Comp})$, it is immediate from \eq{recording-prop} that all intermediate compressed states remain in this image.

These two perspectives are almost identical:
\begin{lemma}[Theorem 3.3 in~\cite{hamoudi2020quantum}]
    Fix a $T$-query algorithm ${\cal A}$ with inter-query unitaries $U_0,\dots,U_T$. Then, for every $t\in[T+1]$, the states 
    \begin{align*}
        &\ket{\psi_t({\cal A})} = U_t{\cal O}U_{t-1}{\cal O}\dots {\cal O}U_0\ket{0}_{\cal WXY}\ket{\sf U}_{\cal I}, &&\ket{\widetilde{\psi}_t({\cal A})} = U_t\widetilde{\cal O}U_{t-1}\widetilde{\cal O}\dots \widetilde{\cal O}U_0\ket{0}_{\cal WXY}\ket{\bot^M}_{\cal I}.
    \end{align*}
    obtained from the standard and compressed oracle models, respectively, satisfy $$\mathsf{Comp}\ket{\psi_t({\cal A})} = \ket{\widetilde{\psi}_t({\cal A})}.$$
\end{lemma}

As in \defin{algorithm}, we write 
\[
\widetilde{\rho}_{\cal I}^t({\cal A}) = \Tr_{\cal WXY}\left[\ket{\widetilde{\psi}_t({\cal A})}\bra{\widetilde{\psi}_t({\cal A})}\right]
\]
for the reduced state of the input register in the compressed oracle model.

The reason for studying the compressed oracle model is that by correctly tracking the $\bot$ symbols in the ${\cal I}$ register of $\ket{\widetilde{\psi}_t({\cal A})}$ with every query, we are able to record what the algorithm has learned about the input. First of all, when applied to a basis state $\ket{w,x,\widehat{p}}_{\cal WXY}\ket{f}_{\cal I}$, $\widetilde{\cal O}$ only changes the value of $f(x)$ stored in register ${\cal I}_x$:
\begin{lemma}[Lemma~4.1 in~\cite{hamoudi2020quantum}]\label{lem:oracle}
    Fix $p\in[N]\setminus\{0\}$ and apply the recording query operator $\widetilde{\cal O}$ to a basis state $\ket{w,x,\widehat{p}}_{\cal WXY}\ket{f}_{\cal I}$ with $\ket{f} \in \mathbb{C}[[N] \cup \{\bot\}]^{M}$. Then the content of the cell register ${\cal I}_x$ transforms as
    \[
    \ket{f(x)}_{{\cal I}_x}\longmapsto \sum_{y\in [N] \cup \{\bot\}} \gamma^{(p)}_{y,f(x)} \ket{y}_{{\cal I}_x},
    \]
    where, for $y,z\in[N]$ with $y\neq z$ 
    \begin{align*}
        &\gamma^{(p)}_{y,\bot}=\frac{\omega_N^{py}}{\sqrt{N}}, &&\gamma^{(p)}_{\bot,\bot}=0, &&\gamma^{(p)}_{\bot,z}=\frac{\omega_N^{pz}}{\sqrt{N}}.
    \end{align*}
    \begin{align*}
        &\gamma^{(p)}_{y,z}=\frac{1 -\omega_N^{py} - \omega_N^{pz}}{N}, &&\gamma^{(p)}_{y,y}=\frac{1+\omega_N^{py}(N-2)}{N}.
    \end{align*}
    
    If $p = 0$, then none of the registers are changed.
\end{lemma}

Since we start with the input initialised to $\ket{\bot^M}_{\cal I}$, \lem{oracle} implies the following consequence, which is the cornerstone of the compressed oracle technique:
\begin{lemma}[Fact~3.2 in~\cite{hamoudi2020quantum}]\label{lem:size}
    For any $T$-query algorithm ${\cal A}$ and every $t\in[T+1]$, the state $\ket{\widetilde{\psi}_t({\cal A})}$ is a linear combination of basis states $\ket{w,x,\widehat{p}}_{\cal WXY}\ket{f}_{\cal I}$ where $f$ contains at most $t$ entries different from $\bot$.
\end{lemma}

\noindent For any $\ket{f} \in \mathbb{C}[[N] \cup \{\bot\}]^{\otimes M}$, we write $\abs{f} = s$ if $f$ contains precisely $s$ entries different from $\bot$.

\subsection{Application to collision finding}\label{sec:collision-progress}

We define the following projectors by giving the computational basis states on which they project:
\begin{itemize}
    \item $\Pi_{\geq1}$ and $\Pi_{=0}$: all basis states $\ket{w,x,\widehat p}_{\cal WXY}\ket f_{\cal I}$ such that $f$ does or does not contain a collision (excluding $\bot$), respectively.
    \item $\Lambda_s$, where $s\geq0$: all basis states $\ket{w,x,\widehat p}_{\cal WXY}\ket f_{\cal I}$ such that $f$ contains exactly $s$ non-$\bot$ entries.
\end{itemize}

Recall from \eq{recording-prop} the projection
\begin{equation}\label{eq:image-iso}
    \Pi_{\sf Comp}:= \mathsf{Comp}\mathsf{Comp}^\dagger = \id_{\cal WXY}\otimes\bigotimes_{x\in[M]}(\id-\proj{\widehat{0}}),
\end{equation}
i.e.~the orthogonal projector onto the image of the isometry $\mathsf{Comp}$ from \eq{comp}. Here each identity inside the tensor product acts on $\mathbb C[[N] \cup \{\bot\}]$.

Let $\eta=(w,x,\widehat p)$ denote a basis value of the algorithm registers, using the chosen computational bases of ${\cal W}$ and ${\cal X}$ and the Fourier basis of ${\cal Y}$. Let $I\subseteq[M]$. We write $\Lambda_{\eta,I} \preceq \Lambda_{\abs{I}}$ for the orthogonal projector onto all basis states with algorithm registers $\eta$ and occupied set $I$. Suppose that $I=\{i_1<\cdots<i_s\}$, and let ${\cal H}_{\eta,I}$ denote the image of $\Lambda_{\eta,I}$. 

\paragraph{Progress measure.}

We follow the standard compressed oracle progress argument for finding collisions~\cite{liu2019finding,hamoudi2020quantum}. Recall that, for $t<T$, $\ket{\widetilde{\psi}_t}$ denotes the state immediately before the $(t+1)$-st query, while $\ket{\widetilde{\psi}_T}$ is the final state, and that
\begin{equation}\label{eq:compress-t}
    \ket{\widetilde{\psi}_t}=\mathsf{Comp}\ket{\psi_t}, \qquad \Pi_{\sf Comp}\ket{\widetilde{\psi}_t}=\ket{\widetilde{\psi}_t},
\end{equation}
and define
\[
\Delta_t:=\norm{\Pi_{\geq1}\ket{\widetilde{\psi}_t}}.
\]

The unitary $U_{t+1}$ applied after the $(t+1)$-st query acts trivially on the database register. Therefore,
\[
\Delta_{t+1}=\norm{\Pi_{\geq1}U_{t+1}\widetilde{\cal O}\ket{\widetilde{\psi}_t}}=\norm{\Pi_{\geq1}\widetilde{\cal O}\ket{\widetilde{\psi}_t}}.
\]
Separating the part of the state that already contains a collision and using that $\widetilde{\cal O}$ is unitary, we obtain
\begin{equation}\label{eq:progress-split}
    \Delta_{t+1} \leq \Delta_t + \norm{\Pi_{\geq1}\widetilde{\cal O}\Pi_{=0}\ket{\widetilde{\psi}_t}}.
\end{equation}

We first need a tool to upper bound the maximal progress made by a single query. This will be accomplished later with the help of the following lemma.
\begin{lemma}\label{lem:one-query}
    Let $\phi$ be a not necessarily normalised collision-free vector, i.e. $\Pi_{=0}\phi = \phi$. Then
    \begin{equation}\label{eq:one-query-collision}
        \norm{\Pi_{\geq1}\widetilde{\cal O}\phi}\leq\frac4{\sqrt N}\bigg(\sum_{\eta}\sum_{I\subseteq[M]}\abs{I}\norm{\Lambda_{\eta,I}\phi}^2\bigg)^{1/2}.
    \end{equation}
\end{lemma}

\begin{proof}

    The recording query operator $\widetilde{\cal O}$ acts as the identity when the ${\cal Y}$ register is in the state $\ket{\widehat{0}}$. This component cannot create a collision from a collision-free database. Hence, we may assume in the remainder for the proof that $\phi$ has no support on the subspace where $\ket{\widehat{p}}_{\cal Y} = \ket{\widehat{0}}_{\cal Y}$.

    For $y\in[N] \cup \{\bot\}$, define
    \[
    \Lambda_{\eta,y,I}:=\proj{y}_{{\cal I}_x}\Lambda_{\eta,I}
    \]
    and decompose
    \[
    \phi=\sum_{y\in[N] \cup \{\bot\}}\sum_{\eta,I}\phi_{\eta,y,I},\qquad \phi_{\eta,y,I}:=\Lambda_{\eta,y,I}\phi.
    \]
    Since $\phi$ is collision-free, $y=\bot$ implies $x\notin I$, while $y\in[N]$ implies $x\in I$.
    
    Fix a block $\phi_{\eta,\bot,I}$ and write $s=\abs{I}$. For a computational-basis state $\ket{\eta,f}$ in its support, \lem{oracle} gives
    \[
    \Pi_{\geq1}\widetilde{\cal O}\ket{\eta,f}=\sum_{j\in I}\gamma^{(p)}_{f(j),\bot}\ket{\eta,f_{x\leftarrow f(j)}},
    \]
    where $f_{x\leftarrow z}$ is obtained from $f$ by setting the value at $x$ equal to $z$. The output states on the right-hand side are pairwise orthogonal, since the values $f(j)$ are pairwise distinct. Moreover, outputs arising from different input basis states in the same block are orthogonal, since the input database is recovered by replacing the value at $x$ by $\bot$. Using $\abs{\gamma^{(p)}_{z,\bot}}\leq1/\sqrt N$ by \lem{oracle}, we obtain
    \begin{equation}\label{eq:one-query-bot}
        \norm{\Pi_{\geq1}\widetilde{\cal O}\phi_{\eta,\bot,I}}^2\leq\frac{s}{N}\norm{\phi_{\eta,\bot,I}}^2.
    \end{equation}
    
    Now fix $y_1\in[N]$. For a computational-basis state $\ket{\eta,f}$ in the support of $\phi_{\eta,y_1,I}$, a collision can only be created if the new value at $x$ agrees with the value at one of the other occupied positions. Hence,
    \[
    \Pi_{\geq1}\widetilde{\cal O}\ket{\eta,f}=\sum_{j\in I\setminus\{x\}}\gamma^{(p)}_{f(j),y_1}\ket{\eta,f_{x\leftarrow f(j)}}.
    \]
    The output states are again pairwise orthogonal. Outputs arising from different input basis states in the same block are also orthogonal, since the input database is recovered by replacing the value at $x$ by the fixed value $y_1$. Using $\abs{\gamma^{(p)}_{z,y_1}}\leq3/N$ for $z\neq y_1$ by \lem{oracle}, we obtain
    \begin{equation}\label{eq:one-query-y}
        \norm{\Pi_{\geq1}\widetilde{\cal O}\phi_{\eta,y_1,I}}^2\leq\frac{9(s-1)}{N^2}\norm{\phi_{\eta,y_1,I}}^2\leq\frac{9s}{N^2}\norm{\phi_{\eta,y_1,I}}^2.
    \end{equation}
    
    For fixed $y$, the outputs belonging to different pairs $(\eta,I)$ are orthogonal. Indeed, the algorithm registers are unchanged by the query. If $y=\bot$, then the final occupied set is $I\cup\{x\}$, from which $I$ can be recovered because $x$ is contained in $\eta$. If $y\in[N]$, then the occupied set remains equal to $I$. It follows from \eq{one-query-bot} that
    \[
    \norm{\Pi_{\geq1}\widetilde{\cal O}\sum_{\eta,I}\phi_{\eta,\bot,I}}\leq\frac1{\sqrt N}\bigg(\sum_{\eta,I}\abs{I}\norm{\phi_{\eta,\bot,I}}^2\bigg)^{1/2}.
    \]
    Similarly, for every fixed $y_1\in[N]$, \eq{one-query-y} gives
    \[
    \norm{\Pi_{\geq1}\widetilde{\cal O}\sum_{\eta,I}\phi_{\eta,y_1,I}}\leq\frac3N\bigg(\sum_{\eta,I}\abs{I}\norm{\phi_{\eta,y_1,I}}^2\bigg)^{1/2}.
    \]
    Using the triangle inequality over $y_1\in[N] \cup \{\bot\}$ and then Cauchy-Schwarz, we obtain
    \[
    \begin{split}
        \norm{\Pi_{\geq1}\widetilde{\cal O}\phi}
        &\leq\frac1{\sqrt N}\bigg(\sum_{\eta,I}\abs{I}\norm{\phi_{\eta,\bot,I}}^2\bigg)^{1/2}+\frac3N\sum_{y_1\in[N]}\bigg(\sum_{\eta,I}\abs{I}\norm{\phi_{\eta,y_1,I}}^2\bigg)^{1/2}\\
        &\leq\frac1{\sqrt N}\bigg(\sum_{\eta,I}\abs{I}\norm{\phi_{\eta,\bot,I}}^2\bigg)^{1/2}+\frac3{\sqrt N}\bigg(\sum_{y_1\in[N]}\sum_{\eta,I}\abs{I}\norm{\phi_{\eta,y_1,I}}^2\bigg)^{1/2}\\
        &\leq\frac4{\sqrt N}\bigg(\sum_{\eta,I}\abs{I}\norm{\Lambda_{\eta,I}\phi}^2\bigg)^{1/2}.\qedhere
    \end{split}
    \]
\end{proof}

\paragraph{From progress to success probability.}

We finally relate $\Delta_T$ to the actual success probability of the algorithm. Recall that $\sf U$ denotes the uniform distribution over functions $f:[M]\rightarrow[N]$ and that for the collision finding problem, the set of valid outputs is defined in \eq{collision-relation} as
\[
{\sf Coll}(f):=\{(x_1,x_2,y)\in[M]^2 \times [N]:x_1<x_2\text{ and }f(x_1)=f(x_2)=y\}.
\]
Let $p_{\rm succ}^{\sf U}$ denote the probability that the algorithm outputs a triple $(x_1,x_2,y)$ in ${\sf Coll}(f)$ on an input $f \sim {\sf U}$.

\begin{lemma}\label{lem:collision-success}
    Let ${\cal A}$ be a $T$-query quantum algorithm. Then,
    \[
    p_{\rm succ}^{\sf U}\leq\bigg(\Delta_T+\sqrt{\frac{2}{N}}\bigg)^2.
    \]
\end{lemma}

\begin{proof}
We apply the compressed oracle readout bound of \cite[Lemma~5]{zhandry2019record} with $k=2$.  Zhandry states the result for a random oracle with range $\{0,1\}^n$. The same proof applies to range $[N]\cong\mathbb Z_N$ by replacing the $n$-fold tensor-product Hadamard transform with the quantum Fourier transform over $\mathbb Z_N$, giving an error term of $\sqrt{2/N}$.

More precisely, we identify an output $(x_1,x_2,y)$ of ${\cal A}$ with the tuple
\[
(x_1,x_2,y_1,y_2) = (x_1,x_2,y,y),
\]
and consider the relation
\[ 
{\cal R} := \{(x_1,x_2,y_1,y_2) \in[M]^2\times[N]^2:: x_1<x_2,\ y_1=y_2\}.
\]
The success event in Zhandry's lemma is therefore precisely the event
\[
x_1<x_2 \qquad\text{and}\qquad f(x_1)=f(x_2)=y,
\]
which occurs with probability $p_{\rm succ}^{\sf U}$.

Now run ${\cal A}$ with the compressed/recording oracle and measure the compressed database after the algorithm produces its output. Let $p_{\rm rec}$ denote the probability that the output $(x_1,x_2,y)$ satisfies $x_1<x_2$ and that the measured database ${\cal D}$ contains
\[
{\cal D}(x_1)={\cal D}(x_2)=y.
\]
By \cite[Lemma~5]{zhandry2019record},
\[
\sqrt{p_{\rm succ}^{\sf U}} \leq \sqrt{p_{\rm rec}} + \sqrt{\frac{2}{N}}.
\]

Whenever the event defining $p_{\rm rec}$ occurs, the recording database contains two distinct positions with the same recorded value, and hence contains a collision. Therefore
\[
p_{\rm rec} \leq \norm{\Pi_{\geq1}\ket{\widetilde{\psi}_T}}^2 = \Delta_T^2.
\]
Consequently,
\[
\sqrt{p_{\rm succ}^{\sf U}} \leq \Delta_T+\sqrt{\frac{2}{N}}.
\]
Squaring both sides proves the claim.
\end{proof}

\paragraph{Known query lower bound.}
Even before incorporating the space restriction, \lem{one-query} already recovers the standard query lower bound. Indeed, by \lem{size}, every database in the support of $\Pi_{=0}\ket{\widetilde{\psi}_t}$ has at most $t$ occupied positions, and hence
\[
\sum_\eta\sum_{I\subseteq[M]}\abs I \norm{\Lambda_{\eta,I}\Pi_{=0}\ket{\widetilde{\psi}_t}}^2 \leq  t\norm{\Pi_{=0}\ket{\widetilde{\psi}_t}}^2 \leq t.
\]
Therefore, \eq{progress-split} and \lem{one-query} give
\[
\Delta_{t+1}\leq\Delta_t+4\sqrt{\frac{t}{N}}, \qquad\text{and thus}\qquad \Delta_T\leq\frac{4}{\sqrt N}\sum_{t=0}^{T-1}\sqrt{t} \leq\frac{4T^{3/2}}{\sqrt N},
\]
where $\Delta_0=0$. Together with \lem{collision-success}, this yields, for $T\geq1$, the standard bound $p_{\rm succ}^{\sf U}=O(T^3/N)$ and hence $T=\Omega(N^{1/3})$ for constant success probability. In the next section, we refine precisely the crude estimate above: label symmetry and the $S$-qubit space bound allow us to replace the factor $t$ by $\min\{t,2S/\log_2N\}$, which yields the optimal time-space tradeoff.

\section{Space bounds in the compressed oracle technique}\label{sec:space-bounds}

\subsection{Symmetry restrictions}\label{sec:symmetry-restrictions}

In our setting, both the uniform input distribution and the collision-finding success condition are invariant under arbitrary permutations of the range labels. Thus, no range label has intrinsic significance, making it natural to expect that treating particular labels asymmetrically offers no advantage. We formalise the corresponding symmetry at the level of the algorithm's reduced input state,
and we call it \emph{label symmetry}.

Although this motivates our restriction, label symmetry cannot presently be imposed without loss of generality under a space bound: symmetrising an arbitrary algorithm by explicitly storing a permutation can require $\Theta(N\log N)$ additional qubits~\cite{ambainis2005new,AmbainisMRR11}, which would obscure the time-space tradeoff. Related invariant-algorithm restrictions have been studied for insertion into an ordered list~\cite{farhi1999} and for ordered search~\cite{carolan2025}.

We first define symmetry for a general group action.
\begin{definition}\label{def:symmetric-algo}
    Let $({\cal H}_{\cal I},\pi)$ be a unitary representation of a finite group $G$. Let ${\cal A}$ be a $T$-query algorithm whose algorithm registers use at most $S$ qubits, and let $\delta$ be an input distribution. We say that $({\cal A},\delta)$ is \emph{$G$-symmetric with respect to $\pi$} if, for every $t\in[T+1]$,
    \[
    \pi(g)\rho_{\cal I}^t({\cal A},\delta)\pi(g)^\dagger     =\rho_{\cal I}^t({\cal A},\delta)\qquad\text{for every }g\in G.
    \]
\end{definition}
In our application, $\mathfrak S_N$ is the symmetric group on the $N$ range labels. We extend every $\sigma\in\mathfrak S_N$ to $[N] \cup \{\bot\}$ by setting $\sigma(\bot)=\bot$, and define its unitary actions on the standard and compressed input registers by
\begin{equation}\label{eq:V-action}
    V_\sigma\ket{f}:=\ket{\sigma\circ f},
\end{equation}
both for $f:[M]\to[N]$ and $f:[M]\to[N] \cup \{\bot\}$. 

We further extend this action to the algorithm's registers by tensoring with the identity; thus, range relabelling acts trivially on all algorithm registers, including the query registers ${\cal X}$ and ${\cal Y}$. Because the uniform state $\ket{\widehat0}$ is fixed by every permutation, these actions satisfy
\begin{equation}\label{eq:comp-equivariant}
    (\id_{\cal WXY} \otimes V_\sigma)\mathsf{Comp} =\mathsf{Comp}(\id_{\cal WXY} \otimes V_\sigma).
\end{equation}
In particular, $(\id_{\cal WXY} \otimes V_\sigma)$ commutes with $\Pi_{\sf Comp}$. Since the action fixes $\bot$ and only permutes the remaining labels, it also commutes with $\Pi_{=0}$, $\Pi_{\geq1}$, $\Lambda_s$, and every $\Lambda_{\eta,I}$.

A direct consequence of \eq{input-gram} and the invariance of ${\sf U}$ under $\mathfrak S_N$-symmetry is the following fact.
\begin{fact}\label{fct:gram}
Under the uniform distribution ${\sf U}$, $\mathfrak S_N$-symmetry is equivalent to the Gram-matrix condition
\begin{equation*}
    \braket{\psi_t^f({\cal A})}{\psi_t^g({\cal A})} =\braket{\psi_t^{\sigma \circ f}({\cal A})}{\psi_t^{\sigma \circ g}({\cal A})}
\end{equation*}
for all $f,g\in[N]^M$, $\sigma\in\mathfrak S_N$, and $t\in[T+1]$.
\end{fact}

Symmetric algorithms satisfy the following key space bound.
\begin{lemma}\label{lem:orbit-span-symm}
    Let ${\cal A}$ be an algorithm whose algorithm registers use at most $S$ qubits, and suppose that $({\cal A},\delta)$ is $G$-symmetric with respect to $\pi$. Then, for every $t\in[T+1]$ and every $\ket{\beta}\in\operatorname{supp}(\rho_{\cal I}^t({\cal A},\delta))$,
    \[
        \dim\operatorname{span}\{\pi(g)\ket{\beta}:g\in G\}\leq2^S.
    \]
\end{lemma}
\begin{proof}
    By assumption, $\rho_{{\cal I}}^t({\cal A},\delta)$ is $G$-invariant, and therefore
    \[
        \operatorname{span}\{\pi(g)\ket{\beta}:g\in G\}
        \subseteq\operatorname{supp}(\rho_{{\cal I}}^t({\cal A},\delta)).
    \]
    Since the joint state $\ket{\psi_t({\cal A},\delta)}$ is pure, its Schmidt rank across the algorithm-input cut gives
    \[
        \operatorname{rank}\bigl(\rho_{{\cal I}}^t({\cal A},\delta)\bigr)
        \leq\dim({\cal H}_{\cal WXY})\leq2^S,
    \]
    which proves the claim.
\end{proof}

\begin{definition}[Label-symmetric algorithm]\label{def:label-symmetric}
    A quantum query algorithm ${\cal A}$ is \emph{label-symmetric} if $({\cal A},{\sf U})$ is $\mathfrak S_N$-symmetric with respect to the range-label representation $\sigma\mapsto V_\sigma$ in \eq{V-action}.
\end{definition}
Every algorithm whose only input-dependent operation asks whether $f(x)=f(x')$ is automatically label-symmetric: such an equality query is unchanged by any permutation of the range labels. We state and prove this formally in \lem{equality-query-label-symmetric}.

Let $\mathsf{Comp}_{\cal I}:=\bigotimes_{x\in[M]}\mathsf{Comp}_x$. The reduced standard and compressed input states obey
\[
\widetilde\rho_{\cal I}^t =\mathsf{Comp}_{\cal I}\rho_{\cal I}^t\mathsf{Comp}_{\cal I}^\dagger.
\]
By \eq{comp-equivariant}, label symmetry therefore transfers to the compressed reduced state without changing its rank. In particular, for every basis value $\eta$ of the algorithm registers, the projected part $(\proj{\eta}_{\cal WXY}\otimes\id_{\cal I})\ket{\widetilde\psi_t}$ has an $\mathfrak S_N$-orbit span of dimension at most $2^S$.

In the rest of this section, we identify the components of a collision-free compressed database whose nonzero vectors have high-dimensional $\mathfrak S_N$-orbits. The orbit bound above then forces every state in the reduced state of the input register of a space-bounded, label-symmetric algorithm to be orthogonal to those components.

\subsection{Collision-free compressed databases}\label{sec:no-col}

To control the progress in \eq{progress-split}, we study two constraints on a database with $s$ recorded entries. Compression requires each recorded register to be orthogonal to $\ket{\widehat{0}}$, while collision-freeness requires the recorded labels to be pairwise distinct. We define the corresponding subspaces and their intersection below.

\begin{definition}
Let $s\geq 1$ be some integer. After fixing the $s$ occupied positions and omitting the $\bot$-entries, the space of \emph{compressed databases} is $W^{\otimes s}$, where 
\begin{equation}\label{eq:W-space}
    W:=\ker(\bra{\widehat{0}}) = \left\{\ket{\psi}\in \mathbb{C}[N]:\braket{\widehat{0}}{\psi}=0\right\}.
\end{equation}

The \emph{collision subspace} of $\mathbb C[N]^{\otimes s}$ is
\begin{equation}\label{eq:C-space}
    C_s := \mathrm{span}\{\ket{y_1,\dots,y_s}: y_i=y_j \text{ for some } i\neq j\} \subset \mathbb{C}[N]^{\otimes s}.
\end{equation}
Its orthogonal complement is the \emph{collision-free subspace} as
\begin{equation}\label{eq:C-perp}
      C_s^\perp := \mathrm{span}\{\ket{y_1,\dots,y_s}: y_1,\dots,y_s \text{ are pairwise distinct}\} \subset \mathbb{C}[N]^{\otimes s}.
\end{equation}

Lastly, the space of \emph{collision-free compressed databases} with $s$ recorded entries is the intersection $C_s^\perp\cap W^{\otimes s}$. 
\end{definition}

We make the intuitive correspondence between the spaces $W^{\otimes s},C_s^\perp$ and (collision-free) compressed database states
precise in \sec{database-size}.

A related subtlety is that the operator $\Pi_{W^{\otimes s}}\Pi_{C_s^\perp}$ is in general not equal to the projection $\Pi_{C_s^\perp\cap W^{\otimes s}}$, because a product of orthogonal projections need not coincide with the orthogonal projection onto the intersection, unless the two projections commute. In our case, that means that projecting a compressed database onto the collision-free subspace need not to preserve the compression constraint. In \lem{proj-diff} we show that, although these two operators are not identical, they are nevertheless close in operator norm.

\subsection{Representations in $C_s^\perp\cap W^{\otimes s}$}\label{sec:database-representations}

In this section we analyse the $\mathfrak S_N$-irreducible representations appearing in $C_s^\perp\cap W^{\otimes s}$. We view $\mathbb{C}[N]$ as the permutation representation of $\mathfrak S_N$ by making it permute the labels:
\[
\sigma\ket{y}=\ket{\sigma(y)} \qquad (\sigma\in\mathfrak S_N,\ y\in[N]).
\]
This induces a diagonal action of $\mathfrak S_N$ on $\mathbb C[N]^{\otimes s}$:
\[
\sigma\ket{y_1,\dots,y_s} = \ket{\sigma(y_1),\dots,\sigma(y_s)}.
\]

We also let $\mathfrak S_s$ act on $\mathbb C[N]^{\otimes s}$ by permuting the tensor factors. Since we later multiply by seminormal idempotents on the right, we use the right-action for this action:
\[
\ket{y_1,\dots,y_s}\cdot\pi = \ket{y_{\pi(1)},\dots,y_{\pi(s)}} \qquad (\pi\in\mathfrak S_s).
\]
The left $\mathfrak S_N$-action and the right $\mathfrak S_s$-action commute.

We show that only Specht modules of shape $\lambda$ whose first row has length exactly $N-s$ occur in $C_s^\perp\cap W^{\otimes s}$, that is, $\lambda=(N-s,\mu)$, when $N$ is large enough. The proof is postponed to the end of this section, and naturally breaks into two intermediate results we will present shortly.
\begin{theorem}\label{thm:kernel}
    Assume $s \geq 1$ and $N\ge 2s$. Then
    \[
    C_s^\perp\cap W^{\otimes s} \cong \bigoplus_{\mu\vdash s} \abs{\mathrm{SYT}(\mu)} S^{(N-s,\mu)}.
    \]
\end{theorem}

The Specht modules appearing in this decomposition have large dimension (if $s$ is large), which will be important later to show that the space $C_s^\perp\cap W^{\otimes s}$ is not reachable by our quantum algorithm if it has limited space at its disposal.
\begin{theorem}[Theorem~E in~\cite{rasala1977minimal}]\label{thm:dim-lower}
    Assume $s \geq 1$ and $N\ge 2s$. Then
    \[
    \min_{\mu\vdash s}\dim S^{(N-s,\mu)} \geq \binom{N}{s}-\binom{N}{s-1}.
    \]
\end{theorem}

By combining these two theorems, we show in the next section that for every nonzero vector in $C_s^\perp\cap W^{\otimes s}$, the dimension of the span of its $\mathfrak S_N$-orbit is large. 

We now set out to prove \thm{kernel} by constructing a linear map $D: C_s^\perp \rightarrow (C_{s-1}^{\perp})^{\oplus s}$, whose kernel is $C_s^\perp\cap W^{\otimes s}$ (see \clm{kernel}), but on the other hand, whose kernel is also isomorphic to $\bigoplus_{\mu\vdash s} \abs{\mathrm{SYT}(\mu)} S^{(N-s,\mu)}$ (see \lem{spectrum}). 

For each $r\in \{1,\dots,s\}$, define the deletion map
\begin{equation}\label{eq:deletion}
    d_r:\mathbb{C}[N]^{\otimes s}\longrightarrow \mathbb{C}[N]^{\otimes (s-1)}, \qquad d_r\ket{y_1,\dots,y_s} = \ket{y_1,\dots,y_{r-1},y_{r+1},\dots,y_s}.
\end{equation}
For $s=1$, we use the convention that the zeroth tensor power is identified with $\mathbb{C}$, so in particular $C_0^\perp:=\mathbb{C}$ and $d_1\ket{y}=1 \in \mathbb{C}$.

By taking the direct sum of these maps for the different values of $r$, we obtain the operator
\begin{equation}\label{eq:D}
    D:= \bigoplus_{r=1}^s d_r\big|_{C_s^\perp}:C_s^\perp\longrightarrow (C_{s-1}^\perp)^{\oplus s}.
\end{equation}

\begin{claim}\label{clm:kernel}
    \[
    \ker D = C_s^\perp \cap W^{\otimes s}.
    \]
\end{claim}
\begin{proof}
    Recall from \eq{W-space} that $W:=\ker(\bra{\widehat{0}})$. Hence,
    \[
    W^{\otimes s} = \bigcap_{r=1}^s \ker\bigg(\id^{\otimes(r-1)}	\otimes \bra{\widehat{0}} \otimes \id^{\otimes(s-r)}\bigg).
    \]
    On the other hand, each contraction in the $r$-th tensor factor is related to the deletion map $d_r$ by
    \[
    \id^{\otimes(r-1)} \otimes \bra{\widehat{0}} \otimes \id^{\otimes(s-r)} = \sqrt{\frac{1}{N}} d_r.
    \]
    Thus, these two maps have the same kernel, and therefore
    \begin{equation}\label{eq:kerdr}
        \bigcap_{r=1}^s \ker d_r = W^{\otimes s}.
    \end{equation}
    Since $D$ is the map obtained by restricting $\bigoplus_{r=1}^s d_r$ to $C_s^\perp$, we have
    \[
    \ker D = C_s^\perp \cap W^{\otimes s}. \qedhere
    \]
\end{proof}

We now look at $D$ from a different perspective. Let
\begin{equation}\label{eq:omega}
    \Omega_s:=\{(y_1,\dots,y_s)\in [N]^s:\ y_1,\dots,y_s \text{ are pairwise distinct}\}.
\end{equation}
Then by \eq{C-perp} it is immediate that the orthonormal basis of $C_s^\perp$ is precisely given by $\bigl\{\ket{y_1,\dots,y_s}:(y_1,\dots,y_s)\in\Omega_s\bigr\}$. Consider the adjoint of the restriction of the deletion map $d_r$ to $C_s^\perp$, which is a linear map from $C_{s-1}^\perp$ to $ C_s^\perp$, acting as 
\begin{equation}\label{eq:adjoint}
    (d_r\big|_{C_s^\perp})^\dagger\ket{y_1,\dots,y_{s-1}} = \sum_{\substack{y\in[N]\\ y\notin\{y_1,\dots,y_{s-1}\}}} \ket{y_1,\dots,y_{r-1},y,y_r,\dots,y_{s-1}}.
\end{equation}

Let $A_{N,s}$ be the adjacency matrix of the \textit{arrangement graph} on the
basis vectors $\{\ket{y_1,\dots,y_s}:(y_1,\dots,y_s)\in\Omega_s\}$ of $C_s^\perp$, where two
basis vectors are adjacent if and only if the corresponding (injective) $s$-tuples differ exactly in one coordinate. So $A_{N,s}$ is a linear operator from $C_s^{\perp}$ to $C_s^{\perp}$, acting as
\begin{equation}\label{eq:adjacency}
    A_{N,s}\ket{y_1,\dots,y_s} = \sum_{r=1}^s \sum_{y\notin\{y_1,\dots,y_s\}}\ket{y_1,\dots,y_{r-1},y,y_{r+1},\dots,y_s}.
\end{equation}

The following claim means that $\ker D$ is the $-s$-eigenspace of $A_{N,s}$, and we can analyse $\ker D$ by studying the spectrum of $A_{N,s}$.
\begin{claim}\label{clm:kernelbis}
    \[
    \ker D=\ker(D^\dagger D)=\ker(A_{N,s}+s\cdot\id).
    \]
\end{claim}
\begin{proof}
We verify that
\begin{equation}\label{eq:laplacian}
    D^{\dagger}D=A_{N,s}+s\cdot\id.
\end{equation}
Indeed, by \eq{adjoint} for each $r$,
\[
(d_r\big|_{C_s^\perp})^\dagger d_r\big|_{C_s^\perp}\ket{y_1,\dots,y_s} = \ket{y_1,\dots,y_s} + \sum_{\substack{\mathbf y'\in\Omega_s\\ \mathbf y' \text{ differs from } (y_1,\dots,y_s)\text{ only at index } r}} 	\ket{\mathbf y'},
\]
so summing over $r\in \{1,\dots,s\}$ counts each neighbour in the arrangement graph exactly once, recovering \eq{adjacency} and additionally 
contributes the diagonal term $s\ket{y_1,\dots,y_s}$. 

Since for every linear map $\ker D=\ker(D^\dagger D)$, we conclude by \eq{laplacian} that 
\[
\ker D=\ker(D^\dagger D)=\ker(A_{N,s}+s\cdot\id). \qedhere
\]
\end{proof}

Chen, Ghorbani, and Wong showed that $-s$ is the least eigenvalue of $A_{N,s}$ for $N\geq2s$, obtained a lower bound on its multiplicity, and conjectured that it is eventually the unique negative eigenvalue~\cite{chen2013cyclic}. Araujo and Bratten subsequently determined the spectrum through character ratios and proved uniqueness under the sufficient condition $N>s(s+1)(s+5)/6$~\cite[Theorem~3.5, Proposition~4.1, and its proof]{araujo2017spectra}. The following lemma determines the entire bottom of the spectrum for every $N\geq2s$.

\begin{restatable}[Arrangement Spectral Gap]{lemma}{Arrangement}\label{lem:spectrum}
    Assume $s \geq 1$ and $N\geq 2s$. The smallest eigenvalue of $A_{N,s}$ on $C_s^\perp$ is $-s$, and its eigenspace is, as an $\mathfrak S_N$-representation,
    \[
    \bigoplus_{\mu\vdash s} \abs{\mathrm{SYT}(\mu)} S^{(N-s,\mu)}.
    \]
    Moreover, the next distinct eigenvalue is exactly $N-(3s-2)$. Consequently, the gap above the least eigenvalue is $N-2s+2$, and $-s$ is the unique negative eigenvalue exactly when $N\geq3s-2$.  
\end{restatable}
\begin{proof}
    See \sec{proof-spectrum}.
\end{proof}

We end this subsection by combining our results together in order to prove \thm{kernel}.
\begin{proof}[Proof of {\thm{kernel}}]
    From \clm{kernel} and \clm{kernelbis}, we have that the kernel of $D$ coincides with both $C_s^\perp\cap W^{\otimes s}$ and the $(-s)$-eigenspace of $A_{N,s}$. Furthermore, \lem{spectrum} tells us that the $(-s)$-eigenspace of $A_{N,s}$ is isomorphic to    
    \[
    \bigoplus_{\mu\vdash s} \abs{\mathrm{SYT}(\mu)} S^{(N-s,\mu)}.
    \]
    
    Thus, in summary, 
    \[
    C_s^\perp\cap W^{\otimes s} = \ker D = \ker (A_{N,s} +s\cdot\id)   \cong \bigoplus_{\mu\vdash s} \abs{\mathrm{SYT}(\mu)} S^{(N-s,\mu)}. \qedhere
    \]
\end{proof}

\subsection{Alternating projections onto $C_s^\perp$ and $W^{\otimes s}$}\label{sec:alternating-projections}

The product $\Pi_{W^{\otimes s}}\Pi_{C_s^\perp}$ is not necessarily equal to the orthogonal projection onto $C_s^\perp\cap W^{\otimes s}$.  In general, the comparison of these three projections is related to the literature on alternating projections (see, e.g.,~\cite{Kayalar1988}). In particular, the operator-norm error is exactly the cosine of the Friedrichs angle between the two subspaces (the angle between their components orthogonal to the intersection).

In our case, we derive a direct upper bound on this error without computing the angle and show, in the following lemma, that these two operators are close.
\begin{lemma}\label{lem:proj-diff}
    Assume $s \geq 1$ and $N\geq2s$. Then
    \begin{equation}\label{eq:proj-diff}
        \norm{\Pi_{C_s^\perp}\Pi_{W^{\otimes s}}-\Pi_{C_s^\perp\cap W^{\otimes s}}}^2 = \norm{\Pi_{W^{\otimes s}}\Pi_{C_s^\perp}-\Pi_{C_s^\perp\cap W^{\otimes s}}}^2 \leq\frac{2s-2}{N}.
    \end{equation}
    In particular, if $u\in \mathbb C[N]^{\otimes s}$ satisfies $u\perp (C_s^\perp\cap W^{\otimes s})$, then 
    \begin{equation*}
        \norm{\Pi_{C_s^\perp}\Pi_{W^{\otimes s}}u}^2\leq\frac{2s-2}{N}\norm{u}^2.
    \end{equation*}
\end{lemma}

\begin{proof}

The equality in \eq{proj-diff} follows from taking the adjoint, so it is sufficient to bound 
\[\norm{\Pi_{W^{\otimes s}}\Pi_{C_s^\perp}-\Pi_{C_s^\perp\cap W^{\otimes s}}}^2.\]

We first consider the case $s=1$. Then,
\[
C_1^\perp=\mathbb C[N], \qquad C_1^\perp\cap W=W.
\]
Consequently,
\[
\Pi_W\Pi_{C_1^\perp} - \Pi_{C_1^\perp\cap W} = \Pi_W-\Pi_W = 0,
\]
and the result is immediate.

Henceforth, assume $s\geq2$. Let $u \in \mathbb{C}[N]^{\otimes s}$ and consider $v := (\Pi_{W^{\otimes s}}\Pi_{C_s^\perp} -\Pi_{C_s^\perp\cap W^{\otimes s}})u$. We first prove that 
\[
v\in W^{\otimes s},\qquad  v\in (C_s^\perp\cap W^{\otimes s})^\perp. 
\]

For the first containment, observe that $v\in W^{\otimes s}$, since the images of both projectors $\Pi_{W^{\otimes s}}$ and $\Pi_{C_s^\perp\cap W^{\otimes s}}$ lie in $W^{\otimes s}$. For the second containment, we prove the more general statement 
\[
\Pi_{C_s^\perp\cap W^{\otimes s}}( \Pi_{W^{\otimes s}}\Pi_{C_s^\perp} -\Pi_{C_s^\perp\cap W^{\otimes s}})=0,
\]
leading to $\Pi_{C_s^\perp\cap W^{\otimes s}}v=0$, that is, $v\in (C_s^\perp\cap W^{\otimes s})^\perp$. Indeed, 
\begin{align*}
    \Pi_{C_s^\perp\cap W^{\otimes s}}( \Pi_{W^{\otimes s}}\Pi_{C_s^\perp} -\Pi_{C_s^\perp\cap W^{\otimes s}})
    &=(\Pi_{C_s^\perp\cap W^{\otimes s}} \Pi_{W^{\otimes s}}\Pi_{C_s^\perp} -\Pi_{C_s^\perp\cap W^{\otimes s}}) \\
    &=(\Pi_{C_s^\perp\cap W^{\otimes s}}  -\Pi_{C_s^\perp\cap W^{\otimes s}})=0,
\end{align*}
where for the second equality we have used that $C_s^\perp\cap W^{\otimes s}$ is a subspace of both $C_s^\perp$ and $W^{\otimes s}$, and therefore  $\Pi_{C_s^\perp\cap W^{\otimes s}} \Pi_{W^{\otimes s}}=\Pi_{C_s^\perp\cap W^{\otimes s}}$ and $\Pi_{C_s^\perp\cap W^{\otimes s}}\Pi_{C_s^\perp}=\Pi_{C_s^\perp\cap W^{\otimes s}}$.

We now decompose $v\in (C_s^\perp\cap W^{\otimes s})^\perp$ into its $C_s^{\perp}$- and $C_s$-components (see \eq{C-space} and \eq{C-perp}). Observe first that, since $C_s\subseteq (C_s^\perp\cap W^{\otimes s})^{\perp}$, the decomposition lies in fact inside $(C_s^\perp\cap W^{\otimes s})^\perp$. So we let $v=z+c$ with $z \in C_s^{\perp} \cap (C_s^\perp\cap W^{\otimes s})^{\perp}$ and $c\in C_s\subseteq (C_s^\perp\cap W^{\otimes s})^{\perp}$.

We relate the norm of $v$ to $u$ as follows, using that $v\in W^{\otimes s}$ and $v \in (C_s^\perp\cap W^{\otimes s})^{\perp}$:
\begin{equation}\label{eq:vzu}
\begin{split}
    \norm{v}^2 &= \left\langle v, \left(\Pi_{W^{\otimes s}}\Pi_{C_s^\perp} - \Pi_{C_s^\perp\cap W^{\otimes s}} \right)u \right\rangle  \\
    &= \left\langle \left(\Pi_{C_s^\perp}\Pi_{W^{\otimes s}} - \Pi_{C_s^\perp\cap W^{\otimes s}} \right) v ,u \right\rangle \\
    &= \left\langle \Pi_{C_s^\perp}v,u \right\rangle =\langle z,u\rangle. 	
\end{split}   
\end{equation}

For now, assume the following claim about $z$ and $c$, which we prove separately below.
\begin{claim}\label{clm:DB}
\begin{equation*}
     (N-2s+2)\norm{z}^2 \leq (2s-2)\norm{c}^2,
\end{equation*}
\end{claim}
The proof then concludes as follows:
\[
\norm{v}^2=\norm{z}^2+\norm{c}^2 \geq \left(1+\frac{N-2s+2}{2s-2}\right)\norm{z}^2 = \frac{N}{2s-2}\norm{z}^2.
\]
Combining this with \eq{vzu} and applying Cauchy-Schwarz gives
\[
\norm{v}^2 	= \abs{\langle z,u\rangle} \leq \norm{z}\norm{u} \leq \sqrt{\frac{2s-2}{N}}\norm{v}\norm{u}.
\]
Dividing by $\norm{v}$ (assuming $v \neq 0$, as otherwise the inequality is trivially true) yields the desired
\[
\norm{v} \leq \sqrt{\frac{2s-2}{N}}\norm{u} \qedhere.
\]
\end{proof}

\begin{proof}[Proof of {\clm{DB}}]
For $r\in\{1,\dots,s\}$, recall the deletion map $d_r$ from \eq{deletion}:
\[
d_r:\mathbb{C}[N]^{\otimes s}\longrightarrow \mathbb{C}[N]^{\otimes (s-1)},\qquad d_r\ket{y_1,\dots,y_s} = \ket{y_1,\dots,y_{r-1},y_{r+1},\dots,y_s}.
\]       
Thanks to \eq{kerdr} in the proof of \clm{kernel},  we have seen that $d_r(W^{\otimes s}) = \{0\}$ for every $r\in\{1,\dots,s\}$. In particular, for our $v \in W^{\otimes s}$,
\begin{equation}\label{eq:dr}
    d_r(z) = -d_r(c) = -\Pi_{C_{s-1}^\perp}d_r(c).
\end{equation}
For the second equality, we used $z\in C_s^\perp$, which implies $d_r(z)\in C_{s-1}^\perp$. By the first equality of \eq{dr}, $d_r(c)\in C_{s-1}^\perp$ as well, so it is invariant under $\Pi_{C_{s-1}^\perp}$ and we may insert this projection.

Taking direct sums of $d_r$ and $\Pi_{C_{s-1}^\perp}d_r$ over $r$ gives the following operators, the first of which was already defined in \eq{D}:
\begin{align}\label{eq:DB}
    &D:= \bigoplus_{r=1}^s d_r\big|_{C_s^\perp}:C_s^\perp\longrightarrow (C_{s-1}^\perp)^{\oplus s}, &&B:= \bigoplus_{r=1}^s \Pi_{C_{s-1}^\perp}d_r\big|_{C_s}:C_s\longrightarrow (C_{s-1}^\perp)^{\oplus s}.
\end{align}
By \eq{dr}, these operators satisfy
\[\norm{Dz}^2 = \norm{Bc}^2.\]

We now study these operators in more detail and prove the following, which will conclude the proof: 
\[
\norm{Dz}^2 \geq (N-2s+2)\norm{z}^2, \qquad \norm{Bc}^2 \leq (2s-2)\norm{c}^2.
\]

We start with the lower bound on $\norm{Dz}^2$. By \eq{laplacian} and \lem{spectrum}, $D^\dagger D=A_{N,s}+s\cdot\id$ is positive semidefinite. Its kernel is $C_s^\perp\cap W^{\otimes s}$ by \clm{kernel}, and its smallest positive eigenvalue is $s+\bigl(N-(3s-2)\bigr)=N-2s+2$. Therefore, on $\bigl(C_s^\perp\cap W^{\otimes s}\bigr)^\perp\cap C_s^\perp$, the operator $D^\dagger D$ is bounded below by $(N-2s+2)\id$. Thus, for any $z\in\bigl(C_s^\perp\cap W^{\otimes s}\bigr)^\perp\cap C_s^\perp$,
\begin{equation}\label{eq:bound-D}
    \norm{Dz}^2 \geq (N-2s+2)\norm{z}^2.
\end{equation}

With respect to the standard orthonormal bases, the matrix entries of $B$ lie in $\{0,1\}$. Starting with a non-injective $s$-tuple in $C_s$, deleting one coordinate through $d_r$ either yields an injective tuple, in which case $\Pi_{C_{s-1}^\perp}$ acts as the identity, or leaves a non-injective tuple, in which case the projection acts as $0$. Since deleting one coordinate from a non-injective $s$-tuple can yield an injective $(s-1)$-tuple for at most two choices of the deleted coordinate, each column of $B$ contains at most $2$ nonzero entries. Each row contains at most $s-1$ nonzero entries, since for fixed $r\in\{1,\dots,s\}$ and fixed injective tuple $(z_1,\dots,z_{s-1})$, the preimages under $d_r$ are obtained by inserting at position $r$ one of the $s-1$ entries $z_1,\dots,z_{s-1}$.
Therefore, for any $c=\sum_j c_j e_j \in C_s$ we have, by Cauchy-Schwarz,
\[
\norm{Bc}^2 = \sum_i\abs{\sum_j B_{ij}c_j}^2 \leq (s-1)\sum_{i,j} B_{ij}\abs{c_j}^2 \leq 2(s-1)\sum_j\abs{c_j}^2 = 2(s-1)\norm{c}^2. \qedhere
\]		        
\end{proof}

\subsection{Space bounds for collision-free database states}\label{sec:database-size}

We now apply the results from the previous subsection to states in the recording query model. In particular, we will now be working both on (tensor products of) $\mathbb C[N]$, as well as $\mathbb C[[N] \cup \{\bot\}]$.

Define
\[
V_{\eta,I}:{\cal H}_{\eta,I}\longrightarrow\mathbb C[N]^{\otimes s},\qquad V_{\eta,I}\ket{\eta,f}=\ket{f(i_1),\dots,f(i_s)}.
\]
This map is unitary, because it maps the computational basis of ${\cal H}_{\eta,I}$ bijectively onto the computational basis of $\mathbb C[N]^{\otimes s}$.

Moreover, the map $V_{\eta,I}$ is $\mathfrak S_N$-equivariant with respect to the action of $\mathfrak S_N$  on ${\cal H}_{\eta,I}$, since for every $\sigma\in\mathfrak S_N$ and every $\ket{\eta,f}\in{\cal H}_{\eta,I}$,
\begin{equation}\label{eq:V-equivariant}
\begin{split}
V_{\eta,I}\sigma\ket{\eta,f} &= V_{\eta,I}\ket{\eta,\sigma\circ f} \\ 
&= \ket{\sigma(f(i_1)),\dots,\sigma(f(i_s))} = \sigma V_{\eta,I}\ket{\eta,f}.
\end{split}
\end{equation}
Here range relabelling fixes $\bot$, so it preserves the occupied set $I$ and hence the subspace ${\cal H}_{\eta,I}$.

The map $V_{\eta,I}$ allows us to define the following commutative diagrams, as stated by the next claim:
\[
\begin{array}{ccc}
    \mathcal H_{\eta,I} & \xrightarrow{\ \Pi_{\sf Comp}\ } & \mathcal H_{\eta,I}
    \\[1ex]
    {\scriptstyle V_{\eta,I}}\downarrow && \downarrow{\scriptstyle V_{\eta,I}} \\[1ex]
    \mathbb C[N]^{\otimes s} & \xrightarrow{\ \Pi_{W^{\otimes s}}\ } & \mathbb C[N]^{\otimes s} 
\end{array},
\qquad\text{and}\qquad
\begin{array}{ccc}
    \mathcal H_{\eta,I} & \xrightarrow{\ \Pi_{=0}\ } & \mathcal H_{\eta,I} \\[1ex]
    {\scriptstyle V_{\eta,I}}\downarrow && \downarrow{\scriptstyle V_{\eta,I}} \\[1ex]
    \mathbb C[N]^{\otimes s} & \xrightarrow{\ \Pi_{C_s^\perp}\ } & \mathbb C[N]^{\otimes s}
\end{array}.
\]

\begin{claim}\label{clm:intertwine}
    On ${\cal H}_{\eta,I}$, the map $V_{\eta,I}$ relates $\Pi_{\sf Comp}$ to $\Pi_{W^{\otimes s}}$ as follows:
    \begin{equation}\label{eq:intertwine}
        V_{\eta,I}\Pi_{\sf Comp}\Lambda_{\eta,I}=\Pi_{W^{\otimes s}}V_{\eta,I}\Lambda_{\eta,I}.
    \end{equation}

Moreover, restricting to collision-free databases corresponds to projecting onto $C_s^\perp$:
\begin{equation}\label{eq:collision-free-intertwine}
    V_{\eta,I}\Pi_{=0}\Lambda_{\eta,I}=\Pi_{C_s^\perp}V_{\eta,I}\Lambda_{\eta,I}.
\end{equation}
\end{claim}
\begin{proof}
It suffices to verify both identities on a computational basis vector $\ket{\eta,f}$. If $\ket{\eta,f}\notin{\cal H}_{\eta,I}$, then $\Lambda_{\eta,I}\ket{\eta,f}=0$, so both sides of both identities vanish. Assume therefore that the occupied set of $f$ is
$I=\{i_1<\cdots<i_s\}$.

First, $f$ is collision-free if and only if $f(i_1),\dots,f(i_s)$ are pairwise distinct. Consequently,
\[
V_{\eta,I}\Pi_{=0}\ket{\eta,f} = 
\begin{cases}
    \ket{f(i_1),\dots,f(i_s)},&
    f(i_1),\dots,f(i_s)\text{ are pairwise distinct},\\
    0,&\text{otherwise},
\end{cases}
\]
which is exactly
\[
\Pi_{C_s^\perp}V_{\eta,I}\ket{\eta,f}.
\]
This proves \eq{collision-free-intertwine}.

For \eq{intertwine}, observe that
\[
(\id-\proj{\widehat0})\ket{\bot}=\ket{\bot}.
\]
Thus, the cells outside $I$ remain equal to $\ket{\bot}$, while on each occupied cell $i_j$ the operator $\Pi_{\sf Comp}$ acts as
$\id-\proj{\widehat0}$. After applying $V_{\eta,I}$, we therefore obtain
\[
V_{\eta,I}\Pi_{\sf Comp}\ket{\eta,f} = \bigotimes_{j=1}^s (\id-\proj{\widehat0})\ket{f(i_j)}.
\]
Since
\[
\Pi_{W^{\otimes s}} = (\id-\proj{\widehat0})^{\otimes s},
\]
the right-hand side equals
\[
\Pi_{W^{\otimes s}}V_{\eta,I}\ket{\eta,f}.
\]
This proves \eq{intertwine}.
\end{proof}

We now state a general fact that we will use for the following theorem.
\begin{fact}\label{fct:orbit-span} 
    Let $V$ be a finite-dimensional representation of ${\mathfrak S}_N$, let $v\in V$, and let $\lambda\vdash N$. Suppose that there is an ${\mathfrak S}_N$-equivariant linear map $B:V\longrightarrow S^\lambda$ such that $B(v)\neq0$. Then 
    \[ 
    \dim\left( \operatorname{span} \{ \sigma v:\sigma\in{\mathfrak S}_N\} \right) \geq \dim S^\lambda. 
    \]
\end{fact}
\begin{proof}
    Let $W=\operatorname{span} \{\sigma v:\sigma\in{\mathfrak S}_N\}\subseteq V$. By the equivariance of $B$,
    \begin{align*}
        B(W) &= \operatorname{span}\left\{B\bigl(\sigma v\bigr): \sigma\in\mathfrak S_N \right\}\\
        &= \operatorname{span} \left\{\sigma B(v):\sigma\in\mathfrak S_N \right\}.
    \end{align*}
    Thus, $B(W)$ is an $\mathfrak S_N$-invariant subspace of $S^\lambda$. Since $S^\lambda$ is irreducible, either $B(W)=\{0\}$ or $B(W)=S^\lambda$. Since $v\in W$ and $B(v)\neq0$, the first case can be excluded. Hence, we obtain the desired result, since $\dim W\geq \dim B(W)$.
\end{proof} 
In particular, if
\[
    V=U\oplus U'
\]
is a direct sum of $\mathfrak S_N$-subrepresentations of $V$, then the projection $B:V\longrightarrow U$ is $\mathfrak S_N$-equivariant. In this case, $B(v)\neq0$ means that $v$ has a nonzero component in $U$ with respect to this decomposition.

\begin{theorem}\label{thm:database-space}
    Assume $N\geq9$, $S\geq\log_2N$, and let $t$ be a nonnegative integer such that $t \leq \sqrt{N}$. Let $\phi\in{\cal H}_{\cal WXY}\otimes{\cal H}_{\cal I}$ be a not necessarily normalised vector satisfying $\Pi_{\sf Comp}\phi=\phi$, and suppose that every database in the support of $\phi$ has at most $t$ entries different from $\bot$. Assume that, for every product-basis value $\eta = (w,x,\widehat{p})$ of the algorithm registers,
    \[
    \dim\left(\operatorname{span}\left\{\sigma(\proj{\eta}_{\cal WXY}\otimes\id_{\cal I})\phi:\sigma\in{\mathfrak S}_N\right\}\right)\leq 2^S.
    \]
    Then, for every $s \in [t+1]$,
    \begin{equation*}
        \norm{\Pi_{=0}\Lambda_s\phi}^2\leq
        \begin{cases}
            \norm{\Lambda_s\phi}^2, & s\leq\dfrac{2S}{\log_2N},\\[2ex]
            \dfrac{2s-2}{N}\norm{\Lambda_s\phi}^2, & s>\dfrac{2S}{\log_2N}.
        \end{cases}
    \end{equation*}
\end{theorem}

\begin{proof}
    Fix $s \in [t+1]$. The first case where $s\leq\dfrac{2S}{\log_2N}$ is immediate, because $\Pi_{=0}$ is an orthogonal projector. 
    
    Suppose now that $s>\frac{2S}{\log_2N}$. Since $s\leq t\leq\sqrt N$ and $N\geq9$, we have $N \geq 2s$ and thus, \thm{kernel}, \thm{dim-lower}, and \lem{proj-diff} are applicable.
         
    For every $\eta$ and every $I\subseteq[M]$ with $\abs{I}=s$, set
    \[
    \phi_{\eta,I}:=\Lambda_{\eta,I}\phi.
    \]
    The rest of the proof consists of proving the statement for each $\phi_{\eta,I}\in {\cal H}_{\eta,I}$, since the Pythagorean identity then proves the theorem after summing over $\eta$ and $I$.
    
    Since $\Pi_{\sf Comp}$ acts independently on each database cell and preserves both $\operatorname{span}\{\ket{\bot}\}$ and $\mathbb C[N]$, it preserves every sector ${\cal H}_{\eta,I}$. Hence, $\Lambda_{\eta,I}$ commutes with $\Pi_{\sf Comp}$ and we have $\Pi_{\sf Comp}\phi_{\eta,I}=\phi_{\eta,I}$.
    
    It follows from \eq{intertwine} that
    \[
    V_{\eta,I}\phi_{\eta,I}\in W^{\otimes s}.
	\]	
    Define
    \[
    \varphi_{\eta,I}:=\Pi_{C_s^\perp\cap W^{\otimes s}}V_{\eta,I}\phi_{\eta,I}.
    \]
    The map $\phi_{\eta,I}\mapsto\varphi_{\eta,I}$ is ${\mathfrak S}_N$-equivariant: relabelling does not change $\eta$ or the occupied set $I$, the map $V_{\eta,I}$ is equivariant by \eq{V-equivariant}, and $C_s^\perp\cap W^{\otimes s}$ is $\mathfrak S_N$-invariant because relabelling preserves pairwise distinctness and fixes $\ket{\widehat 0}$. Therefore,
    \[
    \dim\left(\operatorname{span}\{\sigma\varphi_{\eta,I}:\sigma\in{\mathfrak S}_N\}\right) \leq\dim\left(\operatorname{span}\left\{\sigma(\proj{\eta}_{\cal WXY}\otimes\id_{\cal I})\phi:\sigma\in{\mathfrak S}_N\right\}\right) \leq 2^S.
    \]
    
    By \thm{kernel},
    \[
    C_s^\perp\cap W^{\otimes s}\cong\bigoplus_{\nu\vdash s}\abs{\mathrm{SYT}(\nu)} S^{(N-s,\nu)}.
    \]
    If $\varphi_{\eta,I}\neq0$, then it has a nonzero component in some Specht module $S^{(N-s,\nu)}$. By \fct{orbit-span} and \thm{dim-lower},
    \[
    \dim\left(\operatorname{span}\{\sigma\varphi_{\eta,I}:\sigma\in{\mathfrak S}_N\}\right)\geq\binom Ns-\binom N{s-1}.
    \]
    Since $s>\frac{2S}{\log_2N}\geq2$, we have $s\geq3$, and since $s\leq t\leq\sqrt N$,
    \begin{align*}
        \binom Ns-\binom N{s-1}&=\binom Ns\frac{N-2s+1}{N-s+1}\\
        &\geq\left(\frac Ns\right)^s\frac{s(N-s+1)}{N}\frac{N-2s+1}{N-s+1}\\
        &\geq\left(\frac Ns\right)^ \geq N^{s/2}>2^S.
    \end{align*}
    This is a contradiction, and hence
    \[
    V_{\eta,I}\phi_{\eta,I}\perp C_s^\perp\cap W^{\otimes s}.
    \]
    
    Using \eq{collision-free-intertwine}, the fact that $V_{\eta,I}\phi_{\eta,I}\in W^{\otimes s}$, and \lem{proj-diff}, we obtain
    \[
    \begin{split}
        \norm{\Pi_{=0}\phi_{\eta,I}}^2
        &=\norm{\Pi_{C_s^\perp}V_{\eta,I}\phi_{\eta,I}}^2\\
        &=\norm{\left(\Pi_{C_s^\perp}\Pi_{W^{\otimes s}}-\Pi_{C_s^\perp\cap W^{\otimes s}}\right)V_{\eta,I}\phi_{\eta,I}}^2\\
        &\leq\frac{2s-2}{N}\norm{\phi_{\eta,I}}^2. \qedhere
    \end{split}
    \]
\end{proof}

The following corollary of \thm{database-space} will be useful in our collision finding application in the next section.

\begin{corollary}\label{cor:database-space}
    Under the assumptions of \thm{database-space},
    \begin{equation*}
        \sum_{s=0}^t s\norm{\Pi_{=0}\Lambda_s\phi}^2 \leq \min\left\{t,\frac{2S} {\log_2N}\right\}\norm{\phi}^2.
    \end{equation*}
\end{corollary}

\begin{proof} 
    The proof follows from the inequality \[
    s\norm{\Pi_{=0}\Lambda_s\phi}^2 \leq \min\left\{t,\frac{2S}{\log_2N}\right\}\norm{\Lambda_s\phi}^2. 
    \] 
    Indeed, we can then conclude by 
    summing over $0\leq s\leq t$ and using $\sum_s \norm{\Lambda_s\phi}^2\leq \norm{\phi}^2$ because of
    the orthogonality of the projectors $\Lambda_s$. 
    
    Let us now prove the claimed inequality. For $s\leq\frac{2S}{\log_2N}$, \thm{database-space} gives 
    \[
    s\norm{\Pi_{=0}\Lambda_s\phi}^2 \leq \norm{\Lambda_s\phi}^2 \leq \min\left\{t,\frac{2S}{\log_2N}\right\}\norm{\Lambda_s\phi}^2, 
    \] 
    where we used that $s\leq t$. 
    
    Now suppose that $s>\frac{2S}{\log_2N}$. Then \thm{database-space} gives 
    \[
    s\norm{\Pi_{=0}\Lambda_s\phi}^2 \leq \frac{2s(s-1)}{N}\norm{\Lambda_s\phi}^2 \leq  \min\left\{t,\frac{2S}{\log_2N}\right\} \norm{\Lambda_s\phi}^2,
    \]
    where we first used $s\leq t\leq\sqrt N$ to get $\frac{2s(s-1)}{N}\leq2$ and then $s>\frac{2S}{\log_2N} \geq 2$ and $s\leq t$ to obtain $ 2 \leq \min\left\{t,\frac{2S}{\log_2N}\right\}$. 
\end{proof}

\subsection{Time-Space tradeoff}\label{sec:time-space} 

We now apply \cor{database-space} to bound the progress of the algorithm towards finding a collision. 

\begin{lemma}\label{lem:coll-prog}
    Let ${\cal A}$ be a label-symmetric algorithm. Assume $N\geq9$,$S\geq\log_2N$, and $T \geq 0$. Then
    \[
    \Delta_T\leq\frac{4T}{\sqrt N}\sqrt{\min\left\{T,\frac{2S}{\log_2N}\right\}}.
    \]
\end{lemma}

\begin{proof}
    If $T>\sqrt N$, then, since $S\geq\log_2N$ and $N\geq9$,
    \[
        \min\left\{T,\frac{2S}{\log_2N}\right\} \geq2.
    \]
    Hence the right-hand side of the claimed inequality is greater than $1$, whereas $\Delta_T\leq1$. Thus, the result is immediate in this case.

    It remains to consider the case where $T \leq\sqrt N$. We have $\Delta_0=0$, since the empty database contains no collision. We show that, for every $0\leq t<T$,
    \begin{equation}\label{eq:collision-recurrence}
        \Delta_{t+1}\leq\Delta_t+\frac4{\sqrt N}\sqrt{\min\left\{t,\frac{2S}{\log_2N}\right\}}.
    \end{equation}
    
    By \eq{compress-t}, we have $\Pi_{\sf Comp}\ket{\widetilde{\psi}_t}=\ket{\widetilde{\psi}_t}$, and, by \lem{size}, every database in the support of $\ket{\widetilde{\psi}_t}$ has at most $t$ non-$\bot$ positions. 
    
    We next verify the orbit-span hypothesis of \cor{database-space}. We expand the compressed state in the chosen product basis $\eta$ of the algorithm registers as
    \[
    \ket{\widetilde\psi_t} = \sum_\eta \alpha_{\eta,t}\ket{\eta}_{\cal WXY} \ket{\psi_{\eta,t}}_{\cal I},
    \]
    Then
    \[
    \widetilde\rho_{\cal I}^t = \sum_\eta \abs{\alpha_{\eta,t}}^2\proj{\psi_{\eta,t}},
    \]
    so every $\ket{\psi_{\eta,t}}$ belongs to $\operatorname{supp}(\widetilde\rho_{\cal I}^t$). By label symmetry of $\rho_{\cal I}^t$ and \eq{comp-equivariant}, this support is also invariant under the range-label action. Moreover, since $\mathsf{Comp}_{\cal I}$ is an isometry, we have by the orbit span bound of \lem{orbit-span-symm} that 
    \[
    \operatorname{rank}(\widetilde\rho_{\cal I}^t) = \operatorname{rank}(\rho_{\cal I}^t) \leq2^S.
    \]
    Therefore, for every $\eta$,
    \[
    \dim\operatorname{span}\left\{\sigma(\proj{\eta}_{\cal WXY}\otimes\id_{\cal I})\ket{\widetilde\psi_t} :\sigma\in\mathfrak S_N\right\} \leq 2^S,
    \]
    and we may apply \cor{database-space} to $\ket{\widetilde\psi_t}$, which gives
    \begin{equation}\label{eq:space-progress-bound}
        \sum_{\eta}\sum_{I\subseteq[M]}\abs{I}\norm{\Pi_{=0}\Lambda_{\eta,I}\ket{\widetilde\psi_t}}^2\leq\min\left\{t,\frac{2S}{\log_2N}\right\}\norm{\ket{\widetilde\psi_t}}^2 = \min\left\{t,\frac{2S}{\log_2N}\right\}.
    \end{equation}
    
    The vector $\Pi_{=0}\ket{\widetilde\psi_t}$ is collision-free, so by applying \lem{one-query}, using that $\Pi_{=0}$ commutes with every $\Lambda_{\eta,I}$, and then using
    \eq{space-progress-bound}, we obtain
    \begin{align*}
        \norm{\Pi_{\geq1}\widetilde{\cal O}\Pi_{=0}\ket{\widetilde\psi_t}} &\leq \frac4{\sqrt N} \bigg(\sum_\eta\sum_{I\subseteq[M]}\abs{I}\norm{\Lambda_{\eta,I}\Pi_{=0}\ket{\widetilde\psi_t}}^2\bigg)^{1/2}\\
        &=\frac4{\sqrt N}\bigg(\sum_\eta\sum_{I\subseteq[M]}\abs I\norm{\Pi_{=0}\Lambda_{\eta,I}\ket{\widetilde\psi_t}}^2\bigg)^{1/2} \\
        &\leq \frac4{\sqrt N}\sqrt{\min\left\{t,\frac{2S}{\log_2N}\right\}}.
    \end{align*}
    Combining this with \eq{progress-split} proves \eq{collision-recurrence}:
    \[
    \Delta_T\leq\frac4{\sqrt N}\sum_{t=0}^{T-1}\sqrt{\min\left\{t,\frac{2S}{\log_2N}\right\}}\leq\frac{4T}{\sqrt N}\sqrt{\min\left\{T,\frac{2S}{\log_2N}\right\}}. \qedhere
    \]
\end{proof}

Combining \lem{collision-success} with \lem{coll-prog} gives the following bound on the actual success probability.

\begin{corollary}\label{cor:collision-success}
    Let ${\cal A}$ be a label-symmetric algorithm. Assume $N\geq9$, $S\geq\log_2N$, and $T \geq 0$. Then
    \[
    p_{\rm succ}^{\sf U} \leq \left(\frac{4T}{\sqrt N}\sqrt{\min\left\{T,\frac{2S}{\log_2N}\right\}}+\sqrt{\frac{2}{N}}\right)^2.
    \]
    In particular, for $T \geq 1$
    \[
    p_{\rm succ}^{\sf U} \leq \frac{36T^2}{N} \min\left\{T,\frac{2S}{\log_2N}\right\}.
    \]
\end{corollary}

\begin{proof}
    The first inequality follows immediately from \lem{coll-prog} and \lem{collision-success}. Set
    \[
    m:=\min\left\{T,\frac{2S}{\log_2N}\right\}.
    \]
    Since $T\geq1$ and $S\geq\log_2N$, we have $m\geq1$, and hence
    \[
    4T\sqrt m+\sqrt2\leq6T\sqrt m.
    \]
    Squaring proves the second inequality.
\end{proof}

\begin{theorem}\label{thm:collision-time-space}
    Let $M,N$ be positive integers with $N\geq9$, and let $f:[M]\to[N]$ be chosen uniformly at random. Let ${\cal A}$ be a label-symmetric quantum algorithm that makes $T$ queries to $f$ and uses $S$ qubits of space. If ${\cal A}$ outputs a triple $(x_1,x_2,y) \in [M]^2 \times [N]$ satisfying $x_1 < x_2$ and $f(x_1)=f(x_2) = y$ with probability at least $2/3$, then
    \[
    T=\Omega(N^{1/3}) \qquad\text{and}\qquad T^2S=\Omega(N\log N).
    \]
\end{theorem}

\begin{proof}
    Since the range-value register ${\cal Y}$ has dimension $N$, our definition of space implies
    \[
    S\geq\log_2N.
    \]

    If $T=0$, then the first inequality of \cor{collision-success} gives
    \[
    p_{\rm succ}^{\sf U} \leq \frac{2}{N} < \frac{2}{3},
    \]
    since $N\geq9$. Thus, $T\geq1$.

    By \cor{collision-success} and the assumption that the success probability is at least $2/3$,
    \[
    T^2 \min\left\{T,\frac{2S}{\log_2N}\right\} \geq\frac{N}{54}.
    \]
    Since
    \[
    \min\left\{T,\frac{2S}{\log_2N}\right\} \leq T,\qquad \min\left\{T,\frac{2S}{\log_2N}\right\} \leq\frac{2S}{\log_2N}
    \]
    we obtain
    \[
    T^3\geq\frac{N}{54} \qquad\text{and}\qquad  T^2S \geq \frac{N\log_2N}{108}. \qedhere
    \]
\end{proof}

For the problem of Element Distictness, the task is to decide whether a given function $f: [n] \to [q]$ is injective or not. We instead prove a lower bound for the search version, where the full collision triple $(x_1,x_2,y)$ must be output. The decision and
search versions have the same asymptotic bounded-error quantum query complexity~\cite{aaronson2004QLowerBndCollisionAndElementDistinct,ambainis2004QWalkForElementDist}.

Since a uniformly random function $f:[n]\to[n^2]$ contains a collision with constant probability, we can take the domain size to be $n$ and the range size to be $n^2$ in \thm{collision-time-space} to obtain the following result.

\begin{corollary}\label{cor:search-element-distinctness}
    Let $n\geq9$ be a positive integer, and let $\mathcal A$ be a label-symmetric quantum algorithm such that, for every non-injective function $f:[n]\to[n^2]$, the algorithm outputs a triple $(x_1,x_2,y) \in [n]^2 \times [n^2]$ satisfying $x_1<x_2$ and $f(x_1)=f(x_2)=y$ with probability at least $2/3$. If $\mathcal A$ makes $T$ queries and uses $S$ qubits, then
    \[
    T=\Omega(n^{2/3})\qquad\text{and}\qquad T^2S=\Omega(n^2\log n).
    \]
\end{corollary}

\section{Examples of label-symmetric algorithms}\label{sec:label-symmetric}

We verify that the BHT collision-finding algorithm~\cite{brassard1997collision} and Ambainis's quantum walk for Element Distinctness~\cite{ambainis2004QWalkForElementDist} with their usual list-based implementations are in fact label-symmetric, and we additionally show that label symmetry is guaranteed in the equality-query oracle model. The following criterion will be convenient.

\begin{lemma}\label{lem:label-equivariance}
    Let $\mathcal A$ be a $T$-query algorithm, let $\ket{\psi_t^f(\mathcal A)}$ denote its pure state at the checkpoint after exactly $t$ oracle calls on input $f:[M]\to[N]$, with $t=0$ denoting the initial state, and let $f$ be drawn from the uniform distribution ${\sf U}$ on $[N]^M$. Suppose that, for every $\sigma\in\mathfrak S_N$ and every $t\in\{0,\ldots,T\}$, there is a unitary $R_{\sigma,t}$ on $\mathcal H_{\mathcal{WXY}}$, independent of $f$, such that
    \begin{equation}\label{eq:algorithm-equivariance}
        \ket{\psi_t^{\sigma\circ f}(\mathcal A)}=R_{\sigma,t}\ket{\psi_t^f(\mathcal A)}\qquad\text{for every }f:[M]\to[N].
    \end{equation}
    Then $\mathcal A$ is label-symmetric.
\end{lemma}

\begin{proof}
    For every $f,g:[M]\to[N]$, unitarity of $R_{\sigma,t}$ and \eq{algorithm-equivariance} give
    \[
    \braket{\psi_t^{\sigma\circ g}(\mathcal A)}{\psi_t^{\sigma\circ f}(\mathcal A)}=\braket{\psi_t^g(\mathcal A)}{\psi_t^f(\mathcal A)}.
    \]
    The uniform distribution is invariant under the bijection $f\mapsto\sigma\circ f$, so the Gram-matrix characterization from \fct{gram} proves the claim.
\end{proof}

\subsection{BHT and Ambainis's quantum walk}

We use the oracle from \defin{query-og} in its addition form:
\[
\mathcal O_f\ket{x,y}_{\cal XY}=\ket{x,y+f(x)\bmod N}_{\cal XY}\qquad(x\in[M],\ y\in\mathbb Z_N).
\]

Recall from \eq{V-action} that $V_\sigma\ket f=\ket{\sigma\circ f}$ denotes the relabelling action on the input register. We use the same notation for the permutation unitary $V_\sigma\ket y=\ket{\sigma(y)}$ on a single label valued register and for its tensor powers on several such registers. This is consistent with \eq{V-action}, where $V_\sigma$ acts on all $M$ values of $f$.

We now describe implementations of the BHT algorithm and Ambainis’s quantum walk. Since we are concerned only with query complexity, not time complexity, we omit the time optimal implementation, which sorts the data structure by label value and thereby breaks the necessary label symmetry.

\paragraph{The BHT collision-finding algorithm.} Fix a parameter $r$. A convenient unitary implementation proceeds as follows.

\begin{enumerate}
    \item Prepare a uniform superposition over tuples $K=(x_1 < \cdots <x_r) \subseteq [M]^r$ in the index fields of the table.
    \item Query the indices $x_1,\ldots,x_r$ and store
    \[
    D_f(K):=\bigl((x_i,f(x_i))\bigr)_{i=1}^r.
    \]
    \item Reversibly compute a flag $c_K$ indicating whether $D_f(K)$ contains a collision and, when $c_K=1$, compute the lexicographically least pair $(x_i,x_j)$ with $i<j$ and $f(x_i)=f(x_j)$.
    \item Perform the prescribed fixed number of Grover or amplitude-amplification iterations over $[M]\setminus K$ with marking predicate
    \begin{equation}\label{eq:bht-marking}
        h_f(x):=\mathbf 1\left[\exists i\in[r]\text{ such that }f(x)=f(x_i)\right].
    \end{equation}
    A phase query for $h_f$ loads $f(x)$ into a clean temporary register, computes $h_f(x)$ using equality tests against the stored values, applies the phase, uncomputes the equality tests, and erases the temporary value.
    \item Load the value of the candidate index $x$, reversibly verify that it is marked, and select the least $x_i$ satisfying $f(x)=f(x_i)$. If $c_K=1$, copy the internal collision from Step~3 to the output; otherwise copy the verified pair $(x_i,x)$. Finally, uncompute all flags and measure only the output registers.
\end{enumerate}

\paragraph{Ambainis's quantum walk for Element Distinctness.} Fix a parameter $r$. A list-based implementation of the walk on the Johnson graph $J(M,r)$ proceeds as follows.

\begin{enumerate}
    \item Prepare a uniform superposition over tuples $R=(x_1 < \cdots <x_r) \subseteq [M]^r$ in the index fields of the table.
    \item Query the indices $x_1,\ldots,x_r$ thereby store
    \[
    D_f(R):=\bigl((x_i,f(x_i))\bigr)_{i=1}^r.
    \]
    \item Implement the checking reflection using the marking predicate
    \begin{equation}\label{eq:ambainis-marking}
        m_f(R):=\mathbf 1\left[\exists i<j\text{ such that }f(x_i)=f(x_j)\right].
    \end{equation}
    The predicate is computed and uncomputed using reversible equality tests among the stored values.
    \item Implement the Johnson graph update using labels $x\in R$ and $y\notin R$. With $R':=R\setminus\{x\}\cup\{y\}$, the data update is the map
    \begin{equation}\label{eq:johnson-shift-data}
        \ket{R,x,y,D_f(R)}\longmapsto\ket{R',y,x,D_f(R')}.
    \end{equation}
    Concretely, move the slot $(x,f(x))$ to the update register, erase $f(x)$, replace $x$ by $y$, load $f(y)$, swap the two vertex labels, and apply a reversible index-controlled permutation of complete slots to restore increasing order; any comparison or swap-history workspace is uncomputed before the shift ends.
    \item Apply the quantum walk from the checking reflection, the shift, and input-independent reflections. At the end, select the lexicographically least colliding pair in $D_f(R)$, copy it to the output, uncompute all flags, and measure only the output registers.
\end{enumerate}

\begin{remark}
    Our lower bound assumes that the algorithm must output the collision value $y$, whereas in both the algorithms above, the collision output label $y$ value is already stored in the table. Thus outputting the label requires only a $\lceil\log_2N\rceil$-qubit output field and no additional oracle call. Equivalently, starting from only a collision pair $(x_1,x_2)$ output implementation, one final ordinary query obtains the common value with the same asymptotic query and space complexities.
\end{remark}

\begin{lemma}\label{lem:bht-label-symmetric}
    The fixed-query unitary implementations of the BHT algorithm and Ambainis's quantum walk described above are label-symmetric.
\end{lemma}

\begin{proof}
    Fix $\sigma\in\mathfrak S_N$. First, observe that an input-independent unitary requires no separate symmetry argument at the level of the reduced input state, since it acts only on the algorithm registers. In terms of the sufficient criterion \eq{algorithm-equivariance}, if $\ket{\phi^{\sigma\circ f}}=R_\sigma\ket{\phi^f}$ before an input-independent unitary $U$, then
    \[
    U\ket{\phi^{\sigma\circ f}}=(UR_\sigma U^\dagger)U\ket{\phi^f},
    \]
    and the new relating unitary is still independent of $f$. The only points of interest are are therefore the input-dependent Grover and quantum walk blocks and the \emph{ordinary oracle calls}. Ordinary oracle calls are used to load and erase values in the table, for example in Step~2 in both algorithms.
    
    Whenever a table contains $r$ queried values, $V_\sigma^{\otimes r}$ denotes the application of $V_\sigma$ to its $r$ value fields and the identity on all index and auxiliary registers. By definition,
    \begin{equation}\label{eq:data-list-relabel}
        V_\sigma^{\otimes r}\ket{D_f(K)}=\ket{D_{\sigma\circ f}(K)},\qquad V_\sigma^{\otimes r}\ket{D_f(R)}=\ket{D_{\sigma\circ f}(R)}.
    \end{equation}
    
    \paragraph{BHT.} Let $\mathsf M_f$ denote the Grover marking reflection, where the temporary query and workspace registers have returned to zero. On every reachable basis state it acts as
    \begin{equation}\label{eq:bht-phase-action}
        \mathsf M_f\ket{K,D_f(K),x}=(-1)^{h_f(x)}\ket{K,D_f(K),x},
    \end{equation}
    where $h_f$ is defined in \eq{bht-marking}. Since $\sigma$ is injective,
    \begin{equation}\label{eq:bht-predicate-invariant}
    \begin{split}
        h_{\sigma\circ f}(x)&=\mathbf 1\left[\exists i\in[r]\text{ such that }\sigma(f(x)) =\sigma(f(k_i))\right] \\
        &=\mathbf 1\left[\exists i\in[r]\text{ such that }f(x)=f(k_i)\right]=h_f(x).
    \end{split}
    \end{equation}
    Combining \eq{data-list-relabel}, \eq{bht-phase-action}, and \eq{bht-predicate-invariant} gives, on the reachable subspace,
    \begin{equation}\label{eq:grover-equiv}
        \mathsf M_{\sigma\circ f}V_\sigma^{\otimes r}=V_\sigma^{\otimes r}\mathsf M_f.
    \end{equation}
    Since the rest of the Grover iterate, i.e. the diffusion step, acts trivially on the registers storing labels, \eq{grover-equiv} also holds for the entire Grover iterate.
    
    Lastly, the collision flag and output indices are invariant under $\sigma$, because $f(x_i)=f(x_j)$ if and only if $\sigma(f(x_i))=\sigma(f(x_j))$, whereas the output collision label $y$ is mapped by $V_{\sigma}$. 
    
    \paragraph{Ambainis's quantum walk.} Let $\mathsf C_f$ denote the checking reflection. On every reachable basis state it acts as
    \[
    \mathsf C_f\ket{R,D_f(R)}=(-1)^{m_f(R)}\ket{R,D_f(R)},
    \]
    where $m_f$ is defined in \eq{ambainis-marking}. Since equality of range values is preserved by $\sigma$, we have $m_{\sigma\circ f}(R)=m_f(R)$, and hence, on the reachable walk subspace,
    \begin{equation}\label{eq:walk-check-equivariance}
        \mathsf C_{\sigma\circ f}V_\sigma^{\otimes r}=V_\sigma^{\otimes r}\mathsf C_f.
    \end{equation}
    
    Let $\mathsf S_f$ denote the Johnson graph update. For $x\in R$, $y\notin R$, and $R':=R\setminus\{x\}\cup\{y\}$, \eq{johnson-shift-data} gives
    \begin{equation}\label{eq:walk-shift-equivariance}
    \begin{split}
        \mathsf S_{\sigma\circ f}V_\sigma^{\otimes r}\ket{R,x,y,D_f(R)}&=\mathsf S_{\sigma\circ f}\ket{R,x,y,D_{\sigma\circ f}(R)}\notag\\
        &=\ket{R',y,x,D_{\sigma\circ f}(R')}\notag\\
        &=V_\sigma^{\otimes r}\ket{R',y,x,D_f(R')}\notag\\
        &=V_\sigma^{\otimes r}\mathsf S_f\ket{R,x,y,D_f(R)}.
    \end{split}
    \end{equation}
    Therefore, on the reachable walk subspace,
    \begin{equation}\label{eq:walk-shift-equivariance-operator}
        \mathsf S_{\sigma\circ f}V_\sigma^{\otimes r}=V_\sigma^{\otimes r}\mathsf S_f.
    \end{equation}
    
    Every remaining reflection in the walk acts trivially on the registers storing labels  and therefore it follows from \eq{walk-check-equivariance} that \eq{walk-shift-equivariance-operator} holds for the full quantum walk operator.
    
    The collision flag and the collision output are the same as for BHT.
    
    \paragraph{Ordinary oracle calls.} It remains to check the ordinary queries used to load and erase values. Let $J\ket z=\ket{-z\bmod N}$, and let $\ket d$ contain $\ell$ other stored values. On the reachable subspace,
    \begin{align*}
        \mathcal O_{\sigma\circ f}\bigl(V_\sigma^{\otimes\ell}\ket d\ket{x,0}\bigr) &=(V_\sigma^{\otimes\ell}\otimes V_\sigma) \mathcal O_f\ket d\ket{x,0},\\
        \mathcal O_{\sigma\circ f}\bigl((V_\sigma^{\otimes\ell}\otimes JV_\sigma J) \ket d\ket{x,-f(x)}\bigr)&=(V_\sigma^{\otimes\ell}\otimes\id)\mathcal O_f\ket d\ket{x,-f(x)}.
    \end{align*}
    Thus, loading a label a factor $V_\sigma$ on its target, while erasing a label removes it. Equality tests preserve this action, since $f(x)=f(x_i)$ if and only if $\sigma(f(x))=\sigma(f(x_i))$.
    
    We can therefore conclude that, after each query in the algorithm, we can define a unitary $R_{\sigma,t}$ independent of $f$ satisfying \eq{algorithm-equivariance} and the lemma follows from \lem{label-equivariance}.
\end{proof}

\paragraph{Query and space complexities.} For BHT in the standard case $M=N$, the table is prepared with $r$ queries and each Grover iteration uses $O(1)$ ordinary queries. The standard analysis gives $O(\sqrt{N/r})$ iterations, and hence
\[
T_{\rm BHT}=O\bigg(r+\sqrt{\frac Nr}\bigg)
\]
queries~\cite{brassard1997collision}.

For Ambainis's walk, preparing the initial table uses $r$ queries, checking uses no queries once the table is stored, and each shift uses $O(1)$ queries. The standard walk analysis therefore gives
\[
T_{\rm walk}=O\bigg(\max\left\{r,\frac{M}{\sqrt r}\right\}\bigg),
\]
which becomes $O(M^{2/3})$ for $r=\Theta(M^{2/3})$~\cite[Theorem~4]{ambainis2004QWalkForElementDist}.

For both algorithms, the $r$ explicit index slots use $\Theta\left(r(\log M+\log N)\right)$ qubits. All additional auxiliary registers fit within the same asymptotic bound and therfore both implementations use
\[
\Theta\left(r(\log M+\log N)\right)
\]
qubits. When $M=N^{O(1)}$, this is $\Theta(r\log N)$, so fitting either implementation into $S$ qubits requires $r=O(S/\log N)$.

\subsection{Equality-query algorithms}

Pairwise equality access is a classical restricted query model in which the algorithm may ask only whether two input positions contain identical values.

For $f:[M]\to[N]$, the coherent equality oracle is defined by
\begin{equation}\label{eq:equality-oracle}
    \mathcal E_f\ket{x,x',b}=\ket{x,x',b\oplus\mathbf 1[f(x)=f(x')]}
\end{equation}
for $x,x'\in[M]$ and $b\in\{0,1\}$. We call an algorithm an \emph{equality-query algorithm} if its only input-dependent operation is $\mathcal E_f$. In particular, comparing $f(x)$ with a fixed range label is not an equality query in this sense.

One equality query has a clean implementation using the standard oracle. Load $f(x)$ and $f(x')$ into two clean value registers, compute their equality into $b$, and erase both values using $\mathcal O_f^\dagger=J\mathcal O_fJ$. Thus one equality query uses four ordinary queries and $O(\log N)$ additional qubits; the two temporary value registers are returned to zero and reused.

\begin{lemma}\label{lem:equality-query-label-symmetric}
    Every equality-query algorithm is label-symmetric.
\end{lemma}

\begin{proof}
    For every $\sigma\in\mathfrak S_N$ and $x,x'\in[M]$,
    \[
    (\sigma\circ f)(x)=(\sigma\circ f)(x')\quad\Longleftrightarrow\quad f(x)=f(x'),
    \]
    so $\mathcal E_{\sigma\circ f}=\mathcal E_f$. Since the initial state and every inter-query unitary are independent of $f$, induction gives $\ket{\psi_t^{\sigma\circ f}(\mathcal A)}=\ket{\psi_t^f(\mathcal A)}$ for every $t\in\{0,\ldots,T\}$. Thus \eq{algorithm-equivariance} holds with $R_{\sigma,t}=\id$, and \lem{label-equivariance} proves the claim.
\end{proof}

\subsection{Extension to computations with intermediate measurements}\label{sec:intermediate-measurements}
Following the approach of~\cite{ChungGLQ20}, we extend the compression oracle technique for collision finding to computations with intermediate measurements. To do so, we condition the current state on the history of intermediate measurement outcomes. We express this mathematically using a direct-sum decomposition of the space, where each subspace corresponds to one possible history. We then extend our result to this setting.
Nonetheless, we have to be careful, since our space bound is exploited through the Schmidt rank of the full state in \lem{orbit-span-symm}.

\paragraph{Kraus maps and measurement histories.}

Instead of considering intermediate measurements separately from unitary evolutions, we prefer to consider them together using Kraus maps. By doing so, we do not deviate too much from the unitary evolution formalism.
We replace each unitary map $U$ by a Kraus map $K=(K_k)_k$, where
$K_k$ is a linear operator on the algorithm's registers, satisfying the trace-preserving condition $\sum_k K_k^\dagger K_k=\id$. This replaces the map $\ket{\psi}\mapsto U\ket{\psi}$ by
$$\ket{\psi}\mapsto \frac{1}{\norm{K_k\ket{\psi}}} K_k\ket{\psi} \quad\text{with probability}\quad \norm{K_k\ket{\psi}}^2.$$
For convenience, we retain the unnormalised vector $K_k\ket{\psi}$, whose squared norm gives the probability of this outcome to occur. This preserves linearity and allows us to use the direct-sum notation
\[
\ket\psi\mapsto K\ket\psi=\bigoplus_k K_k\ket\psi.
\]
Observe that this map is an isometry from the original space $H$, to the direct sum $\bigoplus_k H$. The isometry comes from the trace-preserving condition. If a measurement outcome has several Kraus operators, we refine the history to include their indices as well, so that each branch remains pure. The direct sum is bookkeeping for the analysis; it is not an additional coherent workspace. The space bound $S$ includes any classical records retained in the algorithm's memory.

This observation generalises directly to directly to sequences of Kraus maps. For instance, for a fixed input $f$, the state before the first query is $\ket{\psi_0}=U_0\ket0$, where we omit the explicit dependence on $f$. Replacing $U_0$ by $K_0=(K_{0,k})_k$ gives
\[
\ket{\psi_0}=K_0\ket0=\bigoplus_k\ket{\psi_{0,k}},
\qquad \ket{\psi_{0,k}}=K_{0,k}\ket0.
\]
We continue inductively. Just before the $(t+1)$-st query, let the current state be
\[
\ket{\psi_t}=\bigoplus_h\ket{\psi_{t,h}},
\]
where the sum is over histories from the first $t+1$ Kraus maps $K_0,\ldots,K_t$. After the query ${\cal O}_f$ (acting separately on each history branch) and the next Kraus map $K_{t+1}$ (whose restriction to branch $h$ is denoted by $K_{t+1,h}$), each history $h$ is extended to $hk$, giving
\[
\ket{\psi_{t+1}}=
K_{t+1}\mathcal{O}_f\ket{\psi_{t}} =\bigoplus_k K_{t+1,k}\mathcal{O}_f\ket{\psi_{t}} =\bigoplus_{h,k}\ket{\psi_{t+1,hk}},
\qquad\ket{\psi_{t+1,hk}}=K_{t+1,h,k}{\cal O}_f\ket{\psi_{t,h}}.
\]
Here we allow the Kraus family to depend on the preceding history, with $\sum_k K_{t+1,h,k}^\dagger K_{t+1,h,k}=\id$ for every $h$.

We now move to the compressed oracle technique, which we extend to history-dependent states. From now on, $\ket{\psi_{t,h}}$ denotes the joint algorithm-input state under the uniform input distribution, obtained by replacing ${\cal O}_f$ by the purified oracle ${\cal O}$ and adjoining the initial input state $\ket{\sf U}$. Compression acts separately on every branch, that is, on each copy of the space induced by a history branch. Extending $\Pi_{\geq1}$ similarly to each component-wise on the direct sum over all histories of some given length, we can define
\[
\Delta_t:=\norm{\Pi_{\geq1}\ket{\widetilde\psi_t}}
=\sqrt{\sum_h\Delta_{t,h}^2},
\qquad \Delta_{t,h}:=\norm{\Pi_{\geq1}\ket{\widetilde\psi_{t,h}}},
\]
Write $p_{t,h}=\norm{\ket{\widetilde\psi_{t,h}}}^2$ for the probability of history $h$, so that $\sum_h p_{t,h}=1$.

Extending the compressed oracle to the direct sum in the same way, we can bound the evolution of $\Delta_t$. Each Kraus operator acts trivially on the database register and therefore commutes with $\Pi_{\geq1}$.
By Pythagoras and the isometry property of each Kraus family, we get
$$\Delta_{t+1}^2
=\norm{\Pi_{\geq1}K_{t+1}\widetilde{\cal O}\ket{\widetilde\psi_{t}}}^2\\
=\norm{K_{t+1}\Pi_{\geq1}\widetilde{\cal O}\ket{\widetilde\psi_{t}}}^2\\
=\norm{\Pi_{\geq1}\widetilde{\cal O}\ket{\widetilde\psi_{t}}}^2.\\
$$

As in \eq{progress-split}, splitting the state into its collision and collision-free parts and applying the triangle inequality gives
\[
\Delta_{t+1}\leq\Delta_t+\norm{\Pi_{\geq1}\widetilde{\cal O}\Pi_{=0}\ket{\widetilde\psi_t}}.
\]

\paragraph{Symmetry conditioned on the histories.}

To obtain the space-dependent progress bound, we must now adapt the symmetry argument to states conditioned on measurement histories. Both \lem{one-query} and \lem{collision-success} continue to hold branchwise, so the remaining step is to adapt the notion of label symmetry in \sec{symmetry-restrictions}. For this purpose, define the unnormalised reduced input state for a history $h$ after $t$ queries by
\[
\rho_{{\cal I},h}^t:=\Tr_{\cal WXY}\left[\ket{\psi_{t,h}}\bra{\psi_{t,h}}\right],
\]
whereas the global reduced input state is $\rho_{\cal I}^t=\sum_h\rho_{{\cal I},h}^t$. Each branch is pure, but the rank of this sum need not be bounded by the workspace dimension. Requiring full label symmetry in each fixed branch would be sufficient, but would exclude measuring a label and subsequently using its value. We therefore allow the symmetry to fix the classical record retained by the algorithm.

At each checkpoint, separate the algorithm registers into a classical record register ${\cal C}$ and the remaining registers, with $\dim{\cal H}_{\cal C}\leq2^c$. The basis of ${\cal C}$ is indexed by a finite set $\Omega$ on which $\mathfrak S_N$ acts by relabelling the stored values. For a history of measurement outcomes $h$, the corresponding record has the value $z_h\in\Omega$,
so the branch state is of the form $\ket{z_h}_{\cal C}\ket{\phi_{t,h}}_{({\cal A}\setminus{\cal C}){\cal I}}$. Let
\[
G_{z_h}:=\{\sigma\in\mathfrak S_N:\sigma z_h=z_h\}
\]
be the subgroup of $\mathfrak S_N$ fixing this record. We require, at every checkpoint and for every history, that
\begin{equation}\label{eq:history-label-symmetry}
V_\sigma\rho_{{\cal I},h}^t V_\sigma^\dagger=\rho_{{\cal I},h}^t
\qquad\text{for every }\sigma\in G_{z_h}.
\end{equation}
With an empty record, this reduces to full label symmetry in each branch. 

To recover the orbit-span bound, write $W_{t,h}=\operatorname{supp}(\rho_{{\cal I},h}^t)$. Since the record is fixed in the branch, $\dim W_{t,h}\leq2^{S-c}$. By \eq{history-label-symmetry}, this subspace is invariant under $G_{z_h}$, hence
\[
\dim\operatorname{span}\{V_\sigma W_{t,h}:\sigma\in\mathfrak S_N\}
\leq[\mathfrak S_N:G_{z_h}]\dim W_{t,h}
\leq|\Omega|\,2^{S-c}\leq2^S.
\]
Thus every input vector in a branch still has a full $\mathfrak S_N$-orbit span of dimension at most $2^S$, as required by \cor{database-space}. The space used to store the record accounts for the possible relabellings of that record.

We now give a sufficient condition, stated in \eq{covar-hist}, on the measurement and its unitary continuation to preserve \eq{history-label-symmetry}. For a fixed $t$ and history $h$ with corresponding record $z_h$, the unnormalised algorithm state just before applying $K_{t+1,h,k}$ is
\[
\ket{\chi_{t+1,h}^f}:={\cal O}_f\ket{\psi_{t,h}^f}.
\]
Assume that the reduced input state of $N^{-M/2}\sum_f\ket{\chi_{t+1,h}^f}\ket f$ is invariant under $G_{z_h}$, as required by \eq{history-label-symmetry} at this checkpoint. By \eq{input-gram}, there are input-independent workspace unitaries $R_\sigma$ such that for every $\sigma\in G_{z_h}$
\[
\ket{\chi_{t+1,h}^{\sigma\circ f}}=R_\sigma\ket{\chi_{t+1,h}^f}.
\]
For the remainder of this paragraph, abbreviate $K_k:=K_{t+1,h,k}$ and suppress the fixed indices $t,h$. Let $U_{k,s}^f$ denote the unitary continuation of the algorithm, including oracle calls, from outcome $k$ to its $s$-th checkpoint, with $U_{k,0}^f=\id$ immediately after the measurement. It suffices that input-independent workspace unitaries $R_{\sigma,k,s}$ satisfy, on the reachable states, for every $\sigma\in G_{z_h}$
\begin{equation}\label{eq:covar-hist}
U_{\sigma k,s}^{\sigma\circ f}K_{\sigma k}R_\sigma
=R_{\sigma,k,s}U_{k,s}^fK_k,
\end{equation}
both at $s=0$ and at every ordinary-query checkpoint of the continuation.

\paragraph{Examples of covariant measurements.}

We illustrate two examples of measurement histories that satisfy \eq{covar-hist}. Throughout this section, $t,h$ are fixed and we suppress this indices as at the end of the previous subsection. 

For the first example, suppose that the workspace action $R_\sigma$ relabels two label registers simultaneously, acting on their basis states as $\ket{a,b}\mapsto\ket{\sigma(a),\sigma(b)}$. Consider the measurement on these two registers with Kraus operators
\[
K_{\rm eq}=\sum_{y\in[N]}\proj{y,y},
\qquad K_{\rm neq}=\id-K_{\rm eq},
\]
tensored with the identity on the remaining registers. Since the relabelling action preserves equality, we have
\[
R_\sigma K_{\rm eq}R_\sigma^\dagger =\sum_{y\in[N]}\proj{\sigma(y),\sigma(y)} =K_{\rm eq},
\]
and the same holds for $K_{\rm neq}$. Both outcomes are unchanged by relabelling, hence
\[
K_{\sigma k}R_\sigma=K_kR_\sigma=R_\sigma K_k.
\]
This verifies \eq{covar-hist} immediately after the measurement, with $R_{\sigma,k,0}=R_\sigma$. If the subsequent unitary continuation is symmetric, i.e. it satisfies
\[
U_{k,s}^{\sigma\circ f}R_\sigma =R_{\sigma,k,s}U_{k,s}^f
\]
on the postmeasurement states, then multiplying by $K_k$ verifies \eq{covar-hist} at each continuation checkpoint as well. In particular, the measurement itself is already symmetric under relabelling. 

For the second example, the measurement itself is not symmetric under relabelling. Suppose that a query at an address $x_0$ has placed $f(x_0) = y$ in a clean label register. Conditioned on $y$, the unitary continuation of the algorithm is a Grover search over the remaining indices $x\ne x_0$ for another occurrence of the label $y$. Relabelling the input changes the measured value to $\sigma(y)$, but leaves the set of indices searched for unchanged. Thus, the Grover search depends only on which values are equal, not on their actual labels. More explicitly, the measurement projectors $K_y=\proj y$ satisfy
\[
K_{\sigma(y)}R_\sigma=R_\sigma K_y,
\]
where $R_\sigma$ relabels the measured register. The oracle marks the same indices on inputs $(f,y)$ and $(\sigma\circ f,\sigma(y))$, and the Grover diffusion operator is independent of the labels. This gives
\[
U_{\sigma(y),s}^{\sigma\circ f}K_{\sigma(y)}R_\sigma = R_{\sigma,y,s}U_{y,s}^fK_y
\]
on the reachable states, as required by \eq{covar-hist}.

\section{Bottom spectrum of the arrangement graphs}\label{sec:proof-spectrum}

In this section, we prove \lem{spectrum}, which gives a complete description of the bottom of the spectrum of the arrangement graph $A_{N,s}$ (in the regime $N \geq 2s$). Building on the previous spectral analyses of Chen, Ghorbani, and Wong~\cite{chen2013cyclic} and Araujo and Bratten~\cite{araujo2017spectra}, we identify the full $(-s)$-eigenspace and determine the exact next distinct eigenvalue. The proof combines the representation-theoretic decomposition of $C_s^\perp$ with the character-ratio formula for the eigenvalues of $A_{N,s}$. We begin by recalling the required facts about skew diagrams, horizontal strips, and Young's seminormal idempotents.

If $\lambda$ and $\mu$ are partitions with $\lambda_i\le\mu_i$ for all $i$, then we write $\lambda\subseteq\mu$. The \textit{skew diagram} $\mu/\lambda$ is the set of boxes of $Y(\mu)$ not contained in $Y(\lambda)$:
\[
\mu/\lambda:=Y(\mu)\setminus Y(\lambda).
\]
A skew diagram is called a \textit{horizontal strip} if it contains at most one box in each column. We write
\[
\lambda\prec\mu
\]
if $\mu/\lambda$ is a horizontal strip. Equivalently, one can check that $\lambda\prec\mu$ if and only if
\[
\mu_1\ge \lambda_1\ge \mu_2\ge \lambda_2\ge \mu_3\ge \lambda_3\ge\cdots.
\]

We will also use the following standard facts about Young's seminormal idempotents, see, e.g., \cite[Sections~2.2-2.3]{garsia2020young}. For every partition $\lambda\vdash n$ and every standard Young tableau $\mathfrak{t}\in\mathrm{SYT}(\lambda)$, there is a seminormal idempotent $e_{\mathfrak{t}}\in\mathbb C[\mathfrak S_n]$. These idempotents form a complete family of pairwise orthogonal idempotents:
\begin{equation}\label{eq:idempotent}
    e_{\mathfrak{t}}^2=e_{\mathfrak{t}}, \qquad e_{\mathfrak{t}}e_{\mathfrak{t}'}=0\quad(\mathfrak{t}\neq \mathfrak{t}'), \qquad \sum_{\lambda\vdash n}\sum_{\mathfrak{t}\in\mathrm{SYT}(\lambda)}e_{\mathfrak{t}}=1.
\end{equation}

Since every element of $\mathbb C[\mathfrak S_n]$ is a linear combination of permutations, a right action of $\mathfrak S_n$ on $V$ allows such an element to act on $V$ by linearity. Under this action, the map $v\mapsto ve_{\mathfrak t}$ is an idempotent projection onto
\[
Ve_{\mathfrak t} := \{ve_{\mathfrak t}:v\in V\}.
\]
Moreover, \eq{idempotent} gives the direct-sum decomposition
\[
V = \bigoplus_{\lambda\vdash n} \bigoplus_{\mathfrak{t}\in\mathrm{SYT}(\lambda)} Ve_{\mathfrak{t}}.
\]

We are now ready to prove our Arrangement Spectral Gap lemma, that we restate here for clarity.
\Arrangement*
\begin{proof}
    We first consider the case $s=1$. Then $C_1^\perp=\mathbb C[N]$, and $A_{N,1}$ is the adjacency operator of the complete graph on $N$ vertices. Hence
    \[
    A_{N,1}=N\proj{\widehat{0}}-\id.
    \]
    It follows that $A_{N,1}$ has eigenvalue $N-1$ on
    $\operatorname{span}\{\ket{\widehat{0}}\}$ and eigenvalue $-1$
    on
    \[
    W=\ker(\bra{\widehat{0}})
        \cong S^{(N-1,1)}.
    \]
    Since $(1)$ is the unique partition of $1$ and has exactly one
    standard Young tableau, this is precisely the claimed
    $(-1)$-eigenspace. Moreover, the only other eigenvalue is
    \[
    N-1=N-(3\cdot 1-2),
    \]
    so the claimed bound on the remaining eigenvalues also holds.

    Henceforth, assume that $s\geq2$. Fix a partition $\mu\vdash s$ and a tableau $\mathfrak{t} \in \mathrm{SYT}(\mu)$. Recall that $e_{\mathfrak{t}}\in \mathbb C[\mathfrak S_s]$ is the corresponding seminormal idempotent of the standard Young tableau $\mathfrak{t}$, and let
    \[
    C_{s,\mu,\mathfrak{t}}^\perp:=C_s^\perp e_{\mathfrak{t}},
    \]
    meaning
    \begin{equation*}
        C_s^\perp = \bigoplus_{\mu\vdash s}\ \bigoplus_{\mathfrak{t} \in \mathrm{SYT}(\mu)} C_{s,\mu,\mathfrak{t}}^\perp.
    \end{equation*}
    
    In~\cite[Example~2]{nikitin2019decomposition}, it is shown that each of these subspaces $C_{s,\mu,\mathfrak{t}}^\perp$ satisfies
    \begin{equation*}
        C_{s,\mu,\mathfrak{t}}^\perp \cong 
        \bigoplus_{\lambda: \lambda \prec \mu} S^{(N-\abs{\lambda},\lambda)}.
    \end{equation*}
    Here $\lambda=\varnothing$ is allowed. This occurs when $\mu=(s)$, in which case the corresponding summand is understood as $S^{(N-\abs{\varnothing},\varnothing)}=S^{(N)}$. Among these summands, $S^{(N-s,\mu)}$ is the only one whose first row has length exactly $N-s$; every other summand has a strictly longer first row.
    
    By the proof of~\cite[Theorem~3.5]{araujo2017spectra}, more specifically by~\cite[Lemma~3.2, Proposition~3.3, and Theorem~3.4]{araujo2017spectra}, each summand
    $S^{(N-\abs{\lambda},\lambda)}$ (where $\lambda\prec\mu$) is contained in an eigenspace of $A_{N,s}$, and the corresponding eigenvalue is given by
    \[
    E_{\mu, \lambda} = \binom{N}{2}\frac{\chi_{(N-\abs{\lambda},\lambda)}(\tau_N)}{\chi_{(N-\abs{\lambda},\lambda)}(1)} - \binom{s}{2}\frac{\chi_\mu(\tau_s)}{\chi_\mu(1)} - \binom{N-s}{2}, 
    \]
    where $\tau_N$ and $\tau_s$ are transpositions in $\mathfrak S_N$ and $\mathfrak S_s$, respectively. 
    
    For any partition $\alpha\vdash m$, it is known (see, e.g.,~\cite[Eq.~(2.4.9)]{garsia2020young}) that these character ratios can be rewritten as
    \[
    \binom{m}{2}\frac{\chi_\alpha(\tau_m)}{\chi_\alpha(1)} = \sum_{(i,j)\in\alpha}(j-i) =: \gamma(\alpha).
    \]
    Therefore, the eigenvalues of $A_{N,s}$ are given by
    \begin{equation}\label{eq:eigenvalues}
        E_{\mu, \lambda} 	= \gamma((N-\abs{\lambda},\lambda))-\gamma(\mu)-\binom{N-s}{2}.
    \end{equation}
    
    We first compute the value for $\lambda=\mu$, meaning $(N-\abs{\lambda},\lambda) = (N-s,\mu)$. The first row of $(N-s,\mu)$ contributes
    \[
    \sum_{j=1}^{N-s}(j-1)=\binom{N-s}{2}
    \]
    to $\gamma((N-s,\mu))$. The remaining rows are obtained by moving the Young diagram of $\mu$ down by one row, which subtracts $1$ from the quantity $j-i$ for each of the $s$ boxes of $\mu$. Therefore
    \[
    \gamma((N-s,\mu))=\binom{N-s}{2}+\gamma(\mu)-s.
    \]
    Substituting this into the eigenvalue formula gives
    \[
    E_{\mu,\mu}=-s,
    \]
    meaning $S^{(N-s,\mu)}$ lies in the $-s$-eigenspace of $A_{N,s}$.
    
    Now let $\lambda\prec\mu$ with $\lambda\neq\mu$. Then
    \begin{equation}\label{eq:eigenvalue2}
        E_{\mu,\lambda}=E_{\mu,\mu}+\gamma((N-\abs{\lambda},\lambda))-\gamma((N-s,\mu)) = -s + \gamma((N-\abs{\lambda},\lambda))-\gamma((N-s,\mu)).
    \end{equation}
    To compute $E_{\mu,\lambda}$, we therefore only have to compute the difference $\gamma((N-\abs{\lambda},\lambda))-\gamma((N-s,\mu))$. To aid in this, we define the following:
    \[
    \delta_i:=\mu_i-\lambda_i,\qquad q:=\sum_{i\geq 1}\delta_i=s-\abs{\lambda}.
    \]
    Since $\lambda\subseteq\mu$, we have $\delta_i\geq0$, and since $\lambda\neq\mu$, we have $q\geq 1$.
    In terms of these new quantities, the partition $(N-\abs{\lambda},\lambda)$ is obtained from $(N-s,\mu)$ by deleting $\delta_i$ boxes from the end of row $i+1$ for each $i\geq 1$, and adding all of them, i.e.~$q$, to the end of the first row.
    
    The $q$ extra boxes in the first row of $(N-\abs{\lambda},\lambda)$ have a total contribution of 
    \[
    \sum_{a=0}^{q-1}(N-s+a)=q(N-s)+\binom q2
    \]
    to the difference $\gamma((N-\abs{\lambda},\lambda)) - \gamma((N-s,\mu))$. For each $i\geq 1$, the deleted boxes from row $i+1$ have positions
    \[
    (i+1,\mu_i),\; (i+1,\mu_i-1),\; \dots,\; (i+1,\mu_i-\delta_i+1).
    \]
    Since a deleted box $(i+1,j)$ contributes $-(j-(i+1))$ to the difference, the contribution from these deleted boxes to the difference $\gamma((N-\abs{\lambda},\lambda)) - \gamma((N-s,\mu))$ is
    \[
    -\sum_{a=0}^{\delta_i-1}(\mu_i-a-(i+1)) = 	-\delta_i\mu_i+(i+1)\delta_i+\binom{\delta_i}{2}.
    \]
    Combining these added and deleted contributions gives
    \begin{equation}\label{eq:difference}
        \gamma((N-\abs{\lambda},\lambda))-\gamma((N-s,\mu)) = \sum_{i\geq 1}\delta_i(N-s-\mu_i+i+1)+\binom{q}{2}+\sum_{i\geq 1}\binom{\delta_i}{2}.
    \end{equation}
    
    We now lower bound the right-hand side of \eq{difference}. Observe that when $\delta_i>0$ we have $\mu_i>0$. In those cases, since $\mu\vdash s$, we have $\mu_i\leq s-i+1$ and hence
    \begin{equation}\label{eq:lower-lambda}
        N-s-\mu_i+i+1\geq N-s-(s-i+1)+i+1=N-2s+2i\geq N-2s+2.
    \end{equation}
    The binomial terms in \eq{difference} are also nonnegative, so using $q\geq 1$ we obtain
    \[
    \gamma((N-\abs{\lambda},\lambda))-\gamma((N-s,\mu))\geq q(N-2s+2)\geq N-2s+2.
    \]
    
    Substituting this into \eq{eigenvalue2} yields
    \begin{equation}\label{eq:eigenvalue3}
        E_{\mu,\lambda} \geq -s+N-2s+2=N-(3s-2),
    \end{equation}
    meaning that the eigenvalue of $A_{N,s}$ for every $S^{(N-\abs{\lambda},\lambda)}$ with $\lambda \neq \mu$ is at least $N-(3s-2)$, which, by the assumption $N \geq 2s$, is strictly larger than $-s$.  This means that the $(-s)$-eigenspace of $A_{N,s}$ consists solely of the Specht modules $S^{(N-s,\mu)}$ for $\mu \vdash s$. Summing over all $\abs{\mathrm{SYT}(\mu)}$ standard tableaux of shape $\mu$ shows that the $(-s)$-eigenspace of $A_{N,s}$ is
    \[
     \bigoplus_{\mu\vdash s} \abs{\mathrm{SYT}(\mu)} S^{(N-s,\mu)}.
    \]

    The lower bound in \eq{eigenvalue3} is attained for the partitions
    \[
    \mu=(s), \qquad	\lambda=(s-1).
    \]
    Indeed, in that case one can verify that $\lambda\prec\mu$, $q=1$, $\delta_1=1$, and $\delta_i=0$ for all $i\geq2$. Thus, equality holds in \eq{lower-lambda} and the two binomial terms in \eq{difference} become zero.
\end{proof}

\subsection*{Acknowledgements.}
 
The authors thank Dmitry Grinko for helpful discussions and explanations concerning the representation theory of the symmetric group.

This research was supported in part by the French PEPR integrated projects EPIQ (ANR-22-PETQ-0007)
and HQI (ANR-22-PNCQ-0002), and the ERC Advanced Grant PARQ.

\paragraph{AI statement.} The main ideas and proofs in this work were developed by the authors, with the exception of the connection made by ChatGPT~$5.3$ between $C_s^{\perp}\cap W^{\otimes s}$ and arrangement graphs, via the deletion map $D$ defined in \eq{D}. AI tools were also used for supplementary verification of results, editing the manuscript and the discussing of ideas.

\newpage
\bibliographystyle{alpha}
\bibliography{refs}

\end{document}